\documentclass[12pt,letterpaper]{article}

\usepackage[utf8]{inputenc}
\usepackage[english]{babel}
\usepackage{amsmath}
\usepackage{amsfonts}
\usepackage{amssymb}
\usepackage{amsthm}
\usepackage{makeidx}
\usepackage{graphicx}
\usepackage{lmodern}
\usepackage{kpfonts}
\usepackage{color}
\usepackage{xcolor}

\usepackage{hyperref,url}
\usepackage{latexsym}
\usepackage{enumerate}
\usepackage[all]{xy} \CompileMatrices
\usepackage{latexsym,amscd,amssymb,verbatim,hyperref,dsfont}

\theoremstyle{plain}
\newtheorem{theoA}{Theorem}[section]

\newtheorem{corA}[theoA]{Corollairy}

\newtheorem{propA}[theoA]{Proposition}

\theoremstyle{definition}
\newtheorem{deffA}[theoA]{Definition}
\newtheorem{definition-theorem}[theoA]{Definition-Theorem}

\theoremstyle{remark}
\newtheorem{remA}[theoA]{Remark}

\newcommand{\graphx}{X\textit{-graph}}
\newcommand{\graphy}{Y\textit{-graph}}
\newcommand{\chaitin}{\mathcal M} 
\newcommand{\chaitina}{\mathcal M^z} 
\newcommand{\chaitinb}{\mathcal M} 
\newcommand{\dprof}{\mathcal D} 
\newcommand{\Dprof}{\Delta} 
\newcommand{\gen}{\mathcal{G}} 

\newcommand{\rt}{\textrm{RT}} 
\newcommand{\brt}{\tau} 
\newcommand{\U}{\mathcal{U}} 
\newcommand{\soph}{\textrm{Soph}} 
\newcommand{\depth}{\textrm{Depth}} 
\newcommand{\bdepth}{\textrm{Depth}^\bb} 
\newcommand{\spit}{\textrm{Reach}} 

\newcommand{\bb}{\textsf{B}} 
\newcommand{\bbinv}{\bb^{-1}} 

\newcommand{\namecbb}{busy badger} 
\newcommand{\cbb}{\boldsymbol{B}} 
\newcommand{\cbbinv}{\cbb^{-1}} 
\newcommand{\condcbb}[1]{\cbb_{#1}} 
\newcommand{\hd}{H} 
\newcommand{\bs}[1]{\{0,1\}^{#1}} 
\newcommand{\given}{\, \rvert \,} 
\newcommand{\giv}{\rvert} 
\newcommand{\len}[1]{\rvert#1\rvert} 
\newcommand{\halting}[1]{\mathcal{H}^{#1}} 
\newcommand{\finite}[2]{#1_{[#2]}} 
\newcommand{\halts}{\searrow} 
\newcommand{\finom}[1]{\finite{\Omega}{#1}} 

\newcommand{\elog}{\sim} 
\newcommand{\llog}{\lesssim} 
\newcommand{\glog}{\gtrsim} 
\newcommand{\eun}{\asymp} 
\newcommand{\lun}{\curlyeqprec} 
\newcommand{\gun}{\curlyeqsucc} 

\newcommand{\suchthat}{\colon} 
\newcommand{\st}{\colon} 
\newcommand{\N}{\mathbb N} 
\newcommand{\deq}{\stackrel{\text{df}}{=}} 
\newcommand{\et}{~~\textrm{and}~~} 
\newcommand{\disand}{\qquad \textrm{and}\qquad} 
\newcommand{\discom}{\,\textrm{,}\qquad} 

\newcommand{\noi}{\noindent} 

\newcommand{\nit}[1]{\noindent\textit{#1}}

\def \be {\begin{equation}}
\def \ee {\end{equation}}
\def \bes {\begin{equation*}}
\def \ees {\end{equation*}}
\def \bea {\begin{eqnarray}}
\def \eea {\end{eqnarray}}
\def \beas {\begin{eqnarray*}}
\def \eeas {\end{eqnarray*}}

\author{Charles Alexandre B\'edard}

\title{Relativity of Depth and Sophistication} 
\begin{document}
\maketitle
%

\begin{abstract} Logical depth and sophistication are two quantitative measures of the non-trivial organization of an object.
Although apparently different, these measures have been proven equivalent, when the logical depth is renormalized by the busy beaver function.
In this article, the measures are relativized to auxiliary information and re-compared to one another.
The ability of auxiliary information to solve the halting problem introduces a distortion between the measures.
Finally, similar to algorithmic complexity, sophistication and logical depth (renormalized)
each offer a relation between their expression of $(x, y)$, $(x)$ and~$ (y \given x)$.
\end{abstract}

\section{Introduction}

Around us are many objects that are neither completely trivial nor completely random. They conceal patterns and structures, buried under incidental disorganization.
As Bennett~\cite{bennett1995logical} coins it, they ``contain internal evidence of a nontrivial causal history''. 
Such objects are difficult to model and to explain, yet, \emph{interesting}.
%
And interesting itself is the task of formalizing mathematically this very notion. Computability theory has led to the development of algorithmic information theory (AIT) and computational complexity theory, two domains in which formal notions for this ``interestingness'' have been casted.

Embedded in AIT is the approach of nonprobabilistic statistics, proposed by Kolmogorov~\cite{kolmogorov1974talk} in the mid 70's, which attempts to distil the ``concealed patterns and structures'' from the apparent ``incidental disorganization''. As in probabilistic statistics, the mission of this approach is to find the most plausible model that supports the object. 
Such a model is identified to the simplest one that entails a nearly shortest description of the 
object in two parts. The first part describes the model (structures and patterns) and the second part is a canonical specification of the precise object 
 among all of which are consistent with the model (incidental randomness). Kolmogorov pointed out that the description length of such a model is a value of particular interest. Koppel~\cite{koppel1987complexity} (indirectly) referred to this quantity as the \emph{sophistication} of the object, a first notion of interestingness.

 Unlike probabilistic statistics, however, an individual object is considered, dismissing anything else ``it could have been''. It is not hypothesized to be drawn from some \emph{unexplained probabilistic} process; instead, it is supposed to have originated from an \emph{unknown computable} process\footnote{
This justifies the name ``algorithmic statistics'' also used as a synonym of nonprobabilistic statistics.}.
 This assumption goes hand in hand with the physical Church-Turing's thesis, namely, the belief that physical processes can be simulated with arbitrary accuracy by a universal computer. Indeed if the object comes from ``around us'' it has originated from an unknown physical process, whence the aforementioned assumption.

The other approach to quantify interestingness is from a radically different angle, incorporating ideas from computational complexity theory to AIT.
In the seminal paper~\cite{Kolmogorov1965} in which he defines algorithmic complexity, Kolmogorov concludes by mentioning the ``existence of cases in which an object permitting a very simple program, \emph{i.e.}, with very small complexity~$K(x)$, can be restored by short programs only as the result of thoroughly unreal duration''. He then writes of his intention of further studying the topic, but he published nothing later on the subject. More than twenty years later, in the late 80's, Bennett carried the torch. The most plausible causal histories of an object lie in the shortest computable descriptions. If all those descriptions entail a lengthy computation, this signifies a difficult deductive path and hence non-triviality of the object. Its \emph{Logical depth} is then the running time of its most plausible computable description. 

Although many people~\cite{koppel1987complexity, antunes2009sophistication, ay2010effective, bauwens2010computability} had observed connections between (variants of) sophistication and logical depth, it is only recently that they have been identified~\cite{antunes2017sophistication} as the same quantity, when logical depth is renormalized to map the ``thoroughly unreal duration'' back into a number comparable to a program length (\emph{e.g.}, the length of a model description).
In this paper, I analyse further those two apparently different --- but in fact equivalent --- approaches to measure the buried structures of an interesting object. 

Algorithmic complexity satisfies the chain rule, Eq.~\eqref{eqcr}, which \emph{connects} the complexity of a pair~$(x,y)$ the complexity of $x$ and the complexity of $y$ relative to 
$x$.
 \emph{The goal of this paper is to investigate whether sophistication and depth also exhibit such a connection between~$(x,y)$, $(x)$ and~\mbox{$(y \given x)$}.}
The main exploration then regards the \emph{relativity} of depth and sophistication, namely, how the concepts change when the universal computer is supplemented with auxiliary information. I show that when both are relativized, sophistication no longer amounts to the renormalized logical depth (\S \ref{sec:crtp}). Their difference is shown to be a function of the difficulty to materialize the halting information of the auxiliary string (\S \ref{secgap}). I then reach the goal: I demonstrate that the depth (again, the renormalized version) of a pair of objects $(x,y)$ can be expressed as the maximum between the depth of $x$ and the depth of $y$ relative to $x$; sophistication of a pair admits a similar, yet distorted relation (\S \ref{secpair}). Finally, I revisit the so-called antistochastic strings from running time considerations (\S\ref{anholo}).  



 
\section{Preliminaries} \label{secpre}

Established notions of AIT and nonprobabilistic statistics, as well as elementary reformulations and generalizations are presented in this section. For attributions and more details, see Refs.~\cite{Li2008, shenbook2}. 

\subsubsection*{Notation}
I denote~$\N = \{0,1,2, \dots \}$ and~$\bs{*}= \{ \epsilon, 0, 1, 00, \dots \}$. I refer to finite bit strings simply as ``strings''. The first $i$ bits of a (finite or infinite) string~$x$ is denoted~$\finite{x}{i}$. The length of a string~$x$ and the cardinality of a set~$S$ are denoted~$\len x$ and~$|S|$; the context will distinguish the meaning. 
A quantity $Q$ may depend on some parameter $n$. The quantity $O(g(n))$ [$\Omega(g(n))$] denotes a positive function eventually upper bounded [lower bounded] by $c g(n)$, where $c$ is a constant.
I write $Q \eun f(n)$, $Q \lun f(n)$ and $Q \gun f(n)$ if, respectively, $ Q(n) - f(n)= \pm O(1)$, $Q(n) \leq f(n) + O(1)$ and $Q(n) \geq f(n) - O(1)$.
I write $Q \elog f(n)$, $Q \llog f(n)$ and $Q \glog f(n)$ if, respectively, $ Q(n) - f(n)= \pm O(\log n)$, $Q(n) \leq f(n) + O(\log n)$ and $Q(n) \geq f(n) - O(\log n)$.


\subsection{Algorithmic Complexity}

The question of whether --- and if so how --- one can robustly represent objects ``around us'' digitally (\emph{i.e.}, using a finite alphabet) is not simple. It falls in the realm of philosophy of science, not that of coding theory. For a discussion on the topic, see Ref.~\cite{bedard2019emergence}.
Nonetheless, digital objects can easily be encoded in strings, thereby restricting the theory to the latter.  
The algorithmic complexity $K(x)$ of a string $x$ is the length of the shortest program to compute $x$ on a universal computer. 
For a meaningful definition, a model of computation and a universal computer within the model need to be fixed. However,
from the ability of universal computers to simulate one another and, by the Church-Turing thesis, to simulate any computable process, the algorithmic complexity of a string is independent of the fixed universal computer, up to an additive constant. 
In this sense, the algorithmic complexity can then be viewed as a universal and absolute quantity of information --- or randomness --- in a string.

Chaitin~\cite{Chaitin1975} defines a similar model in which, the universal computer $\U$ is fixed to be a \emph{self-delimiting} Turing machine, \emph{i.e.}, it has a read-only one-way input tape and some work tapes. When the computation begins, a program~$p$ occupies the input tape and an auxiliary string~$z$ occupies a designated work tape. The computation succeeds only if the machine reaches a halting state while its read head is scanning the rightmost bit of~$p$, but no further. This forces the program to contain within itself the information about its own length. A successful computation is denoted by~$\halts$ and~$\U(p,z)$ is then defined to be the string displayed on the work tape at halting. Self-delimitation ensures that for any~$z$ the set~$\{q \st \U(q,z) \halts \}$ is a \emph{prefix-free} set of strings, namely, no member of which is a prefix of another. When no auxiliary information is provided,~$z$ is simply set to~$\epsilon$, and $\U(p,\epsilon)$ is abbreviated to~$\U(p)$.

The (prefix)  \emph{algorithmic complexity} is defined with respect to the above universal computer $\U$ as
\bes
K(x) \deq \min_p \{|p| \st \U(p) = x\} \,,
\ees
and its \emph{conditional} counterpart as
\bes
K(x\given z) \deq \min_p \{|p| \st \U(p,z) = x\} \,.
\ees
Multiple strings can be encoded into a single one via a computable bijection
$(x_1, x_2, \dots, x_n) \mapsto \langle x_1, x_2, \dots x_n \rangle 
$
uniformly defined for any~$n$.
The complexity of multiple strings is thus naturally defined as~$K(\langle x_1, x_2, \dots, x_n\rangle)$.

Let~$x^*$ and $(x\giv z)^*$ 
be the%
\footnote{In the case of multiple programs of minimal length, the fastest trumps.} 
 shortest programs that computes~$x$ with $\epsilon$ and with $z$ as auxiliary information, respectively. 
\begin{remA} \label{remxstar}
Observe that 
\begin{center} 
\begin{picture}(120,40)(-45,0)
\put(0,0){\framebox(30,30){$O(1)$}} 
\put(-23,15){\vector(1,0){20}}
\put(-43,5){\makebox(20,20){$x^*$}} 
\put(33,23){\vector(1,0){20}}
\put(33,7){\vector(1,0){20}}
\put(53,13){\makebox(20,20){$x$}} 
\put(58,-3){\makebox(20,20){$K(x)$}} 
\end{picture}\disand
\begin{picture}(120,40)(-45,0)
\put(0,0){\framebox(30,30){$O(1)$}} 
\put(-23,23){\vector(1,0){20}}
\put(-23,7){\vector(1,0){20}}
\put(-43,13){\makebox(20,20){$x$}} 
\put(-50,-3){\makebox(20,20){$K(x)$}} 
\put(33,15){\vector(1,0){20}}
\put(55,5){\makebox(20,20){$x^*$}} 
\end{picture}~~,
\end{center}
where the diagrams represent that the output(s) can be computed from the input(s) and a $O(1)$ advice. 
Indeed, $K(x)$ and $x$ can be computed from $x^*$ by measuring its length before executing it. And $x^*$ can be determined by a parallel execution of programs of length~$K(x)$, until $x$ is produced.
\end{remA}

A very important relation is the \emph{chain rule},
\be \label{eqcr}
K(x,y) \eun K(x) + K(y \given x^*) \,,
\ee
as it entails a symmetric notion of \emph{mutual information}, so defined as 
\bes
I(x : y) \deq K(y) - K(y \given x^*) \,.
\ees
The ``$\lun$'' side of Equation~\eqref{eqcr} is easily observed, as one way to compute~$\langle x,y \rangle$ is to copy and then execute $x^*$, which can then serve as an auxiliary string to $(y \giv x^*)^*$. At this stage, $\langle x,y \rangle$ can be computed. The ``$\gun$'' side, harder to prove, states that the previous procedure to compute $\langle x, y \rangle$ is nearly optimal in terms of program length.

Observe that by the information equivalence of $x^*$ and $\langle x,K(x)\rangle$, Remark~\ref{remxstar},~$K(y \given x^*) \eun K(y \given x, K(x))$. This is convenient to write the relativized chain rule as
\bes
K(x,y \given z) \eun K(x \given z) + K(y \given x, K(x \given z))\,.
\ees

\subsubsection*{Halting Information} \label{sechalting}

To determine whether, for a given~$p$,~$\U(p)$ is a halting computation or not is an undecidable task. The \emph{halting problem} is perhaps the most famous of computability theory. It can perfectly be framed in AIT, and even, better quantified. 

As suggested by Turing~\cite{turing37}, the halting problem can be encoded into bits. The most straightforward way of doing so is to define the infinite string~$\halting{}$ whose~$i$-th bit is~$1$ if and only if the $i$-th program, in lexicographic order, halts. I denote~$\halting{\leq j}$ the first~$2^{j+1}-1$ bits of~$\halting{}$, which encode the solution to the halting problem for all programs of length~$\leq j$. 
Such a representation of the halting problem is highly redundant, since the same information can be given in much fewer bits. In fact, together with~$j$, the number $\omega_j$ of programs of length $\leq j$ that halt suffices, because one can recover~$\halting{\leq j}$ by running all programs no longer than $j$ in parallel until $\omega_j$ of them have halted. 

%
%
A more elaborate way of encoding the halting problem is through Chaitin's \emph{halting probability} \cite{Chaitin1975} defined as
\bes
\Omega = \sum_{p:\,\U(p) \halts} 2^{-|p|} \,.
\ees
Since the set of halting program is prefix-free, Kraft inequality implies that the sum converges to a number smaller than~$1$. If a program is given to the reference machine~$\U$ with bits picked at random, then the probability that the computation ever halts is~$\Omega$. The first~$j$ bits of~$\Omega$, denoted~$\finite{\Omega}{j}$, can be used to compute~$\halting{\leq j}$,

\begin{center} 
\begin{picture}(120,35)(-45,0)
\put(0,0){\framebox(30,30){$O(1)$}} 
\put(-23,15){\vector(1,0){20}}
\put(-48,5){\makebox(20,20){$\finite{\Omega}{j}$}} 
\put(33,15){\vector(1,0){20}}
\put(53,5){\makebox(25,20){$\halting{\leq j}$}} 
\end{picture} \,.
\end{center}
This is done by running all programs in a dovetailed fashion, and adding~$2^{-|p|}$ to a sum~$M$ (initially set to~$0$) whenever a program~$p$ halts. When the first $j$ bits of the sum stabilize to the first $j$ bits of $\Omega$, \emph{i.e.}, $\finite{M}{j}=\finite{\Omega}{j}$, then no program of length $\leq j$ will ever halt, since such an additional contribution to the sum would contradict the value of~$\Omega$. This process is said to \emph{lower semi-compute}~$\Omega$, since it always returns smaller numbers than~$\Omega$ and they converge to it in the limit of infinite time.

$\Omega$ is an example of an \emph{incomressible} string, namely that~$K(\finom{j}) \gun j$. This can be proved from a \emph{Berry paradox} argument: the ability of $\finom{j}$ to compute $\halting{\leq j}$ also endows it with the ability to produce~$\zeta$, the first string in lexicographic order with complexity $>j$. However, such a computation of~$\zeta$ from~$\finom j$ is only consistent if~$K(\finom{j}) \gun j$. Moreover, as any string of length~$j$,~$\finom j$ has (prefix) complexity $\lun j + K(j)$. Hence,
\bes
j \lun K(\finom j) \lun j + K(j) \,.
\ees


\subsection{Nonprobabilistic Statistics}

Before overviewing the algorithmic treatment of statistics, 
I introduce elementary concepts and notations about subsets of~$\N^2$. They will be useful for illustrational purposes, conciseness of notation and most importantly to unify different definitions under the same umbrella.

\subsubsection*{The Help of $\N^2$}\label{secn2}

A set $\Psi\subseteq \N^2$ is \emph{upwards closed} [resp. \emph{rightwards closed}] if
\bes
(i, \psi) \in \Psi ~\implies ~\forall k,~~ (i, \psi + k) \in \Psi ~~~~~~[\text{resp. } (i + k, \psi) \in \Psi] \,.
\ees
A~\emph{profile} is an upwards and rightwards closed subset of $\N^2$.  The~$L^\infty$-metric endows $\N^2$ with a distance. The \emph{distance} between~$(a_1,a_2)$ and~$(b_1,b_2)$ is given by~$\max (|a_1-b_1|\,, |a_2-b_2|)$.

Let~$\Psi$ be a profile. Its \emph{boundary}~$\partial \Psi$ is the 
subset at distance unity of some point outside of~$\Psi$, \emph{i.e}, each point in~$\partial \Psi$ has at least one of its~$8$ neighbours outside of~$\Psi$.
The \emph{$X$-graph} of 
~$\Psi$ is
\bes
\graphx(\Psi) \deq \{(i,\psi) \in \partial \Psi \st (i,\psi-1) \notin \Psi \} \,.
\ees
It is the graph of some function~$\psi(i)$ represented as usual by the~$Y$ versus~$X$ axes.
The \emph{$Y$-graph} of 
~$\Psi$ is analogously defined as
\bes
\graphy(\Psi) \deq \{(i,\psi) \in \partial \Psi \st (i-1,\psi) \notin \Psi \} \,,
\ees
and is the graph of some function~$i(\psi)$ unusually represented by the~$X$ versus~$Y$ axes. See Figure~\ref{figxgraph}. 

\begin{figure}
	\begin{center}
		\includegraphics[
		trim={2cm 3cm 2cm 6cm},clip, width=8cm]{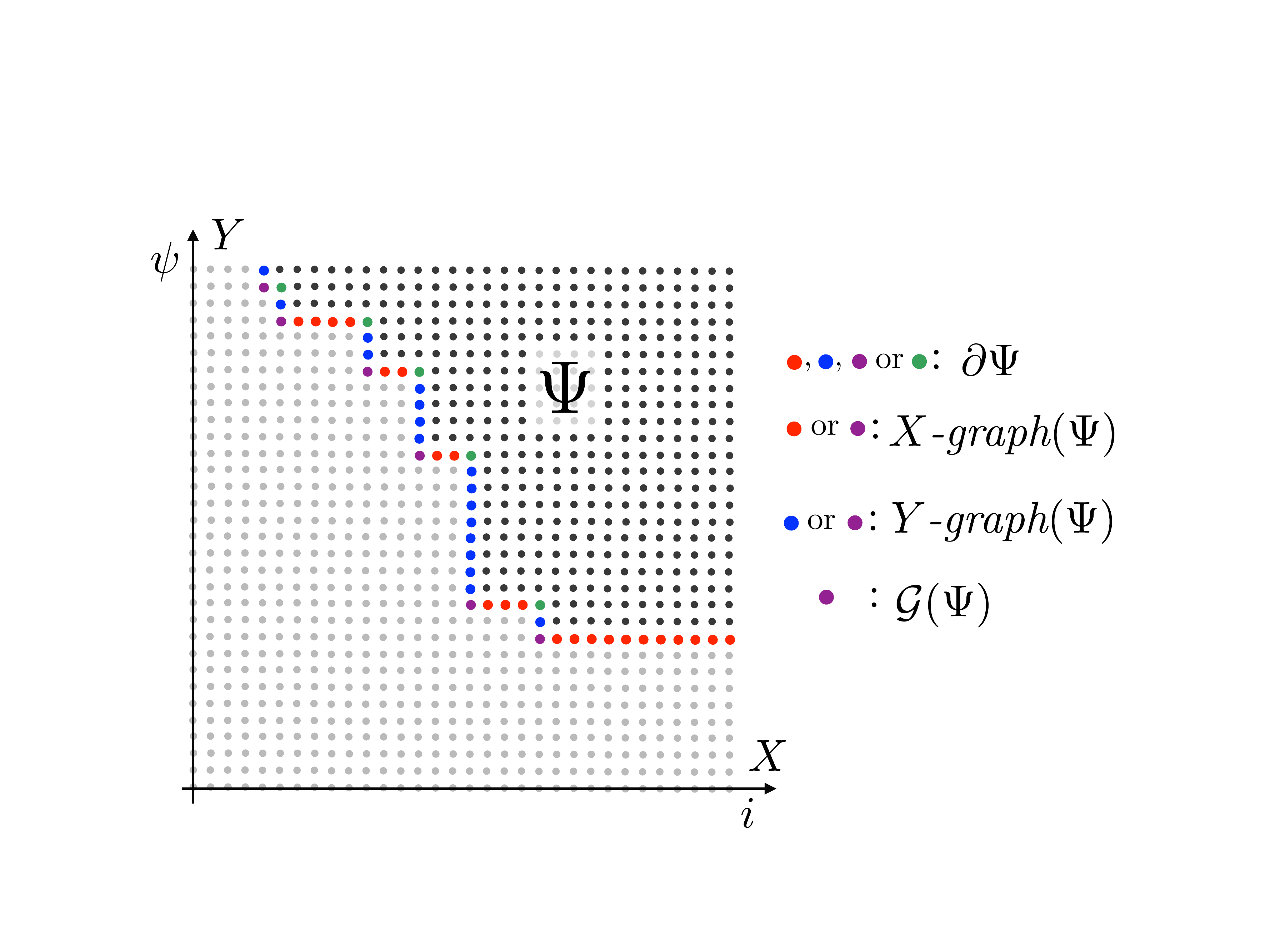}
	\end{center}
\caption{The anatomy of a profile $\Psi \subseteq \N^2$.}
\label{figxgraph}
\end{figure}

\begin{remA}\label{remex}
(Let~$\Psi$ be upwards and rightwards closed.)
Both functions~$\psi (i)$ and~$i(\psi)$, represented respectively by the $X$-\emph{graph} and the $Y$-\emph{graph}, are non-increasing. These functions are in general noninvertible, but they are as close as they can get from being each other's inverse, specifically, 
\bes
\psi(i') = \psi' \implies i(\psi') \leq i' \disand i(\psi') = i' \implies \psi(i') \leq \psi' \,.
\ees
\end{remA}
%
A set $G \subseteq \N^2$ is said to \emph{generate}~$\Psi$ if the upwards and rightwards closure of $G$ gives $\Psi$. Such a closure is understood to be $\{(i',\psi')\in \N^2 \st \exists (i, \psi) \in G ~i \leq i' \et \psi \leq \psi' \}$. Of a particular interest is the minimal such set. The \emph{Generator set} of~$\Psi$ is defined as
\bes
\gen(\Psi) \deq \graphx(\Psi) ~\cap~ \graphy(\Psi) \,.
\ees
It corresponds to the convex corners of~$\Psi$, namely, the corners that have more neighbours outside than inside~$\Psi$. 
%

The \emph{sum} of two profiles $\Psi$ and $\Phi$ is defined as
\bes
\Psi + \Phi \deq \{(i,\psi+\phi) \st (i,\psi) \in \Psi \et (i,\phi) \in \Phi \} \,.
\ees
The $\varepsilon$\emph{-neighbourhood} of $\Phi$ includes all points at a distance $\leq \varepsilon$ of each of its points, hence enlarging the boundary. 
$\Psi$ is $\varepsilon$\emph{-close} to $\Phi$ if it is contained in an $\varepsilon$\emph{-neighbourhood} of $\Phi$. 
\begin{remA}\label{remclose}
(Let~$\Psi$ and~$\Phi$ be upwards and rightwards closed.) $\Psi$ is $\varepsilon$-close to $\Phi$ 
\begin{itemize}
\item[(i)]\label{remclosegen} if and only if~$\gen(\Psi)$ it is contained in an $\varepsilon$\emph{-neighbourhood} of $\Phi$
\item[(ii)]\label{remclosegraph} if and only if $\psi(i) + \varepsilon \geq \phi(i+\varepsilon)$,
\end{itemize}
where~$\psi$ and~$\phi$ are the functions represented by the respective $X$-graphs.
\end{remA}

I denote $\Psi \eun \Phi$ or $\Psi \elog \Phi$ if $\Psi$ and $\Phi$ are both $O(1)$-close or $O(\log n)$-close to one another, respectively.
Those relations find their usefulness in the \emph{two-dimensionality of the approximation}, which cannot be expressed so concisely, for example, by the $X$-graphs.


\subsubsection*{Quantifying ``Good'' Models}

For a review of the field of nonprobabilistic statistics, see Ref.~\cite{vereshchagin2017algorithmic}. 

A finite set~$S$ that contains a string~$x$ is an \emph{algorithmic statistic} of~$x$. It is also called a \emph{model} of~$x$, since it puts together strings that share common properties with~$x$, precisely those that define~$S$. Opposing qualities are expected of a good model. On the one hand, the model should be simple, tending to minimize~$K(S)$. The latter is the length of the shortest program that computes an encoding of the lexicographical ordering of the elements of~$S$ and halts. On the other hand, the canonical description of~$x$ \emph{via the model} should also be minimized. In the case of finite sets as models, such a description amounts to describing first~$S$ and then specifying~$x\in S$ by some canonical encoding, for instance, by giving its index~$i^x_S$ in a lexicographical ordering of the elements of~$S$. 

More precisely, each model~$S\ni x$ entails a \emph{two-part description} of~$x$. The first part consists of describing the model by its shortest program~$S^*$ (of length~$K(S)$) and the second part singles out~$x$ in~$S$, thanks to its index~$i_S^x$ (of~length~$\log |S|$). This second part is known as the \emph{data-to-model code}, but really, it should be called the \emph{model-to-data code}. This means that 
\bes \label{eq2part}
D(S^*, i^x_S) \deq \alpha S^* i_S^x
\ees
is a self-delimiting program that computes~$x$, where the prefix $\alpha$ is a fixed program (of length $O(1)$) which ensures the correct execution of the two-part description. Note that the second part of the code does not need any additional prefix for self-delimitation, since its length $|i_S^x| = \lceil \log |S| \rceil$ can be computed (by~$\alpha$) from~$S^*$. The length of the two-part description is therefore given by
\bes
\len{D(S^*, i^x_S)} = K(S) + \log|S| + \len\alpha  \,.
\ees

The tradeoff between the simplicity of the model and the length of its corresponding two-part description can be expressed by a profile on~$\N^2$: for each $S \ni x$, a dot can be marked at the coordinate~$(K(S), K(S) + \log|S| + \len\alpha)$. The upwards and rightwards closure of those dots yields what I call the \emph{description profile},
\bes
\Lambda_x = \{ (i,\lambda) \st \exists S\ni x\,,~ i \leq K(S) \et K(S) + \log|S| + \len\alpha \leq \lambda \} \,.
\ees
%
The~$\graphx$ of $\Lambda_x$ represents what is known~\cite{vv2002} as the \emph{constrained minimum description length function}
\bes
\lambda_x(i)=\min_{S\ni x}\{K(S) + \log|S| + \len\alpha \st K(S)\leq i  \} \,.
\ees
%
%
For~$i$ large enough,~$\lambda_x(i)$ reaches values close to~$K(x)$. In the worse case, this is achieved for~$i\eun K(x)$ as witnessed by the model~$\{ x\}$. A model~$S$ that entails a two-part description essentially as short as the shortest program is called~\emph{sufficient}. Kolmogorov pointed out that a sufficient model~$S_0$ of \emph{minimal} complexity describes all the structure of~$x$, or in Vit\'anyi's words~\cite{vitanyi2006meaningful}, its ``meaningful information'', but not more. The remaining information~$i^x_{S_0}$ is the incidental or random part of~$x$. The complexity~$K(S_0)$ of a minimal sufficient statistics is now known
 as the~\emph{sophistication} of~$x$. For a precise definition, one needs to clarify what is meant by ``reaches values close to~$K(x)$'', which introduces a resolution parameter~$c$, 
\bes \label{eqdefsoph}
\soph_c(x) \deq \min_{S \ni x} \{K(S) \st K(S) + \log|S| + |\alpha| \leq K(x) + c \} \,.
\ees
Although the sophistication of a string~$x$ is intuitively thought to be the value of~$\soph_c(x)$ for a resolution~$c$ as small as possible, it is meaningful to view~$\soph_c(x)$ as a function of~$c$ since it allows to connect sophistication with the description profile~$\Lambda_x$. 
First, one translates~$\Lambda_x$ down on the $Y$ axis by~$K(x)$ to define
\beas
\Delta_x &\deq& \Lambda_x - (0,K(x)) \\
&=& \{ (i, c) \st \exists S\ni x,~K(S)\leq i \et K(S) + \log|S| + |\alpha| \leq K(x) + c  \}\,.
\eeas 
The $Y$-graph of~$\Delta_x$ is obtained by minimizing the first coordinate, with the second coordinate fixed, yielding~$(\soph_c(x),c)$.

\subsubsection*{Robustness of the Method}

The method used to arrive at a definition of sophistication may appear somewhat arbitrary. Among the different model-selection principles, why minimizing the two-part description? And why imposing finite sets as a model class? Each of these issues have been specifically addressed and the method shows robustness since different model-selection principles and different model classes yields essentially the same measure of sophistication. 

In the method presented here, the quality opposed to the simplicity of the model was the minimality of the two-part description, known as the \emph{minimum description length principle}. The trade-off between those qualities is expressed by the function $\lambda_x(i)$ from which sophistication was read out. Another quality of a model that opposes its simplicity is guided by the \emph{maximum likelihood principle}, which favours the models with as few elements as possible. This trade-off is displayed by the \emph{constrained maximum likelihood function}, 
\bes
h_x(i) = \min_{S\ni x} \{\log |S| \st K(S) \leq i\}\,.
\ees 
This is Kolmogorov's original~\cite{kolmogorov1974talk} \emph{structure function}. Another principle is to minimize the \emph{randomness deficiency}, valuing models~$S$ in which~$x$ is most typical. This defines the function
\bes
\beta_x(i) = \min_{S\ni x} \{\log |S| - K(x \given S) \st K(S)\leq i\}\,,
\ees 
since the lack of typicality is measured by how far from the data-to-model code is the shortest program for computing~$x$ given~$S$.


Importantly, Vereshchagin and Vit\'anyi~\cite{vv2002} showed that the three functions~$\lambda_x$,~$h_x$ and~$\beta_x$ encode the same information, since they are all connected to each other by affine transformations (within logarithmic precision). In particular, the minimal value at which~$\lambda_x(i)$ reaches close to $K(x)$, that is, the sophistication, can be defined alternatively from the maximum likelihood or the randomness deficiency principles. In this paper, the attention is restricted to $\lambda_x$, or more specifically, to its corresponging description profile~$\Lambda_x$. 

The other critique that can be formulated about the path used to define sophistication is the lack of generality of finite sets as a model class. In fact, some people~\cite{gell1996information, gacs2001algorithmic} have generalized the model class to computable probability distributions, possibly infinite. The complexity of the model then becomes that of the distribution and the length of the data-to-model code is then given by the Shannon-Fano code. The constrained minimum description length function, analogous to~$\lambda_x$, is then expressed in terms of these quantities, and again the value at which the function reaches close to~$K(x)$ is identified. Gell-Mann and Lloyd called it~\emph{effective complexity}~\cite{gell1996information}. Yet one more model class possibly even more general is given by total functions\footnote{In fact, the term sophistication was coined by Koppel as he was grasping the idea through total functions as a model class.}, where again, two part-descriptions are analogously defined.

Vit\'anyi~\cite{vitanyi2006meaningful} showed that whether the model class is fixed to finite sets, computable distributions or total functions, the respective description profiles would be close to one another, underlining again the robustness of sophistication, and the sufficiency of finite sets as model class. 

Finally, a very important result of algorithmic statistics states that the description profile~$\Lambda_x$ can essentially take all possible shapes.

\begin{theoA}[All shapes are possible \cite{vv2002}]\label{thmallshapes}
Let~$k\leq n$. Let $\gen$ be some set of points that generates a profile~$\mathcal T$ by upwards and rightwards closure in such a way that $(0, n) \in T$ and~$(k,k) \in \partial T$. Then there exists a string~$x$ of complexity~$k + O(\log n) + K(\gen)$ and length~$n + O(\log n) + K(\gen)$ whose description profile~$\Lambda_x $ is $O(\log n) + K(\gen)$-close to $\mathcal{T}$.
\end{theoA}

%

\subsection{Logical Depth and Time-Bounded Complexity}

One of the most beautiful surprises of algorithmic statistics is that its core concepts are directly related to running-time considerations. 

Hereinafter, $\rt(p)$ stands for the \emph{running time} of $p$, which is the number of computation steps that $\U$ executes on input $p$ before reaching a halting state. If the computation uses auxiliary information~$z$, then I denote $\rt(p,z)$ the running time of the computation $\U(p,z)$. 

An object~$x$ is \emph{deep} if most of its algorithmic probability corresponds to slow computations. The gist of this idea is captured by Bennett's second tentative definition~\cite{bennett1995logical} of \emph{logical depth}, with significance parameter $c$:
\beas
\depth_c(x) \deq \min_{p\colon \U(p)=x} \{\rt (p) \st |p|\leq K(x) + c \} \,.
\eeas

Running times can be very large, especially when interested by the deepest strings of a fixed length. 
The \emph{inverse busy beaver function} renormalizes those astronomical running times back into numbers of size comparable to program length. 

\begin{deffA} \label{def:btime}
The \emph{busy beaver} is a function  $\bb : \N \to \N$ defined by
\bes
\bb(n) \deq \max~\{\rt(p) \st  \U(p) \searrow ~\text{and}~ |p| \leq n  \}\,. 
\ees
\end{deffA}
\noi It is the maximal finite running time of a program of $n$ bits or less. 
Its inverse~$\bbinv(N)$ is then defined as the length of the shortest program that eventually halts after at least~$N$ steps:
\bes
\bbinv(N) = \min \{ |p| \suchthat \rt(p) \geq N \text{ (but finite)} \}\,. 
\ees
As a convenient shortcut, one can measure time right away in busy beaver units by defining the \emph{busy running time}~$\brt(p)$ of a program~$p$ as
\bes
\brt(p) \deq \bbinv (\rt (p))\,.
\ees
Deploying this definition, if~$p$ has a busy running time~$\brt(p)=d$, it means that there is a program of length~$d$, but none of length less than~$d$, that halts after~$p$.

\begin{deffA}\label{def:bdepth}
The \emph{busy beaver depth} of $x$, at significance level $c$ is defined here like in Ref.~\cite{antunes2017sophistication}, but with prefix instead of plain complexity:
\bes
\bdepth_c(x) 
\deq \min_{p\colon \U(p)=x}~\{\brt(p) \st |p| \leq K(x) + c \}\,.
\ees
\noi It amounts to the inverse busy beaver of the logical depth\footnote{The definition of logical depth on which Bennett settled in~Ref.\cite{bennett1995logical} imposes the condition $K(p) \geq |p| - c$ instead of~$|p| \leq K(x) + c$. It has been shown \cite{antunes2017sophistication} that in the plain complexity setting, the inverse busy beaver renormalization of such a definition of logical depth is~$O(1)$ close to the plain complexity counterpart of Def.~\ref{def:bdepth}, up to $O(1)$ precision also in the significance parameter.}. 
\end{deffA}



Related to logical depth is the concept of~\emph{time-bounded complexity}, 
%
\bes
K^t(x ) \deq \min_{p \st \U(p) = x}~\{|p| \st \rt(p) \leq t \} \,.
\ees
The notion was already mentioned in the conclusions of Kolmogorov's seminal paper~\cite{Kolmogorov1965}, as a proposed tool to ``study the relationship between the necessary complexity of a program and its permissible difficulty~$t$''. 
The quoted relationship can be explored through the \emph{time profile} $\mathcal L_x$, generated by the coordinates~$(\tau(p), \len p)$ for each program~$p$ that computes~$x$. 
Written differently,
\bes
\mathcal{L}_x = \{(i,\ell) \st \exists p~~ \U(p) = x \,,~ |p| \leq \ell \et \brt(p) \leq i  \} \,.
\ees
\noi Observe that
\be
\mathcal{L}_x = \{(i,\ell) \st K^{B(i)}(x) \leq \ell \} \,,
\ee
so $(i, K^{B(i)})$ is the $\graphx$ of $\mathcal L_x$. By a process analogous to the reading out of sophistication from the description profile~$\Lambda_x$, the busy beaver depth can be expressed from the time profile~$\mathcal L_x$. To do so, 
%
define
\beas \label{eqdepthprofile}
\dprof_x &\deq& \mathcal L_x - (0,K(x)) \\
&=& 
 \{ (i, c) \st \exists p,~\U(p)=x,~\brt(p)\leq i \et \len{p} \leq K(x) + c  \} \,.
\eeas
The $Y$-graph is obtained by minimizing the first coordinate, with second coordinate fixed, yielding~$(\bdepth_c(x),c)$.

%


The following remarkable result connects the description and time profiles, and so sophistication and depth.
\begin{theoA}[\cite{antunes2017sophistication, bauwens2010computability}]\label{thmsophdepth}
For all $x$,
\bes
\mathcal L_x \elog \Lambda_x \qquad \text{and so} \qquad \dprof_x \elog \Dprof_x \,.
\ees
\end{theoA}

\subsection{Definitions Relativized}

The previous definitions capture properties of a fixed bit string~$x$. The same definitions also hold if one reads~$x$ as an encoding~\mbox{$\langle x',y' \rangle$} of a pair of strings.
The main intention of this paper is to study the properties of description and time profiles when they are relativized\footnote{In this paper, I use ``relativized by'' in the same sense as ``conditional to''.} by some auxiliary information~$z$. Here, I straightforwardly extend the definitions to a conditional counterpart.

The \emph{conditional description and time profiles} are respectively
\beas
\Lambda_{y\giv z} &=& \{(i,\lambda) \st \exists S\ni y~~ i \leq K(S \given z) \et K(S \given z) + \log|S| + \len\alpha \leq \lambda \} \\
\mathcal{L}_{y \giv z} &=& \{(i,\ell) \st \exists p~~ \U(p,z) = x \,,~ |p| \leq \ell \et \brt(p,z) \leq i  \} \,.
\eeas
%
%
Sophistication, busy beaver depth and time-bounded complexity also have a straightforward conditional analogues:
\beas
\soph_c(y\given z) &\deq& \min_{S \ni y} \{K(S\given z) \st K(S\given z) + \log|S| + |\alpha| \leq K(y \given z) + c \}\,, \\
\bdepth_c(y \given z) &\deq& \min_{p\colon \U(p,z)=y}~\{\brt(p,z) \st |p| \leq K(y\given z) + c \}\,, \\
K^{\bb(i)}(y\given z) &\deq&  \min_{p \st \U(p,z) = y}~\{|p| \st \rt(p,z) \leq \bb(i) \} \,.
\eeas
One can again translate profiles
\bes
\Dprof_{y \giv z} \deq \Lambda_{y \giv z} - (0, K(y \given z)) \disand \mathcal D_{y \giv z} \deq \mathcal L_{y \giv z} - (0, K(y \given z))
\ees
and verify that the definitions are consistent with 
\beas
\graphy(\Dprof_{y\giv z}) &=& (\soph_c(y\given z), c) \\
\graphy(\dprof_{y\giv z}) &=& (\bdepth_c(y\given z), c) \\
\graphx(\mathcal L_{y\giv z}) &=& (i, K^{\bb(i)}(y\given z)) \\
\graphx(\Lambda_{y\giv z}) &=& (i, \lambda_{y\giv z}(i)) \,.
\eeas
%


\section{Chain Rules for Profiles} \label{sec:crtp}

One of the most important relations in AIT is the chain rule, eq.~\eqref{eqcr}, for algorithmic complexity. Without it, the ``IT'' in ``AIT'' would be a misnomer, because like in Shannon's theory of information, the chain rule is precisely what entails a symmetric measure of information. A chain rule for time and description profiles would make it possible to express depth and sophistication of a pairs in terms of their single string and conditional version.
%

\subsection{A Chain Rule for Time Profiles}
In this section, I show that the chain rule is carried over by the time profiles within logarithmic resolution, namely, that the following holds
\be \label{eqsickrel}
\mathcal L_{x,y} \elog \mathcal L_x + \mathcal L_{y | x} \,.
\ee
The above relation accounts for two ``profile inequalities''. Each of which is treated independently in Proposition~\ref{propcotefacile} and Proposition~\ref{prophardside}, because they hold with different error bounds.
%
 In both propositions, the strategy is the same: I follow the lines of Longpr\'e's analysis~\cite{longpre1986resource} of the chain rule for time-bounded complexity, but where time is measured in the busy beaver scale and where programs are required to be self-delimited.
\begin{propA}\label{propcotefacile}
For all strings~$x$ and~$y$ of length $\leq n$ and for all~$i\geq\bbinv(n)$, 
\bes
K^{\bb(i')}(x,y) \lun K^{\bb(i)}(x) + K^{\bb(i)}(y\given x),
\ees
where~$i' = i + O(1)$.
\end{propA}
%
\begin{proof}
Let~$p$ and~$q$ be the respective witnesses of~$K^{\bb(i)}(x)$ and $K^{\bb(i)}(y\given x)$. Then $rpq$ is a self-delimiting program for $\langle x, y \rangle$, where $r$ is a constant-size routine that implements the following.
First, $p$ is executed, producing $x$, which is copied before being given as a ressource to~$q$. Thereupon,~$q$ is executed, yielding~$y$,  
and the pair $\langle x, y \rangle$ is computed. 
The running time of the whole computation is $2\bb(i) + O(n) \lun O(\bb(i)) \lun \bb(i + O(1))$.
\end{proof}

\begin{propA} \label{prophardside}
For all strings~$x$ and~$y$ of length $\leq n$ and for all~$i\geq\bbinv(n)$, 
\bes
K^{\bb(i')}(x) + K^{\bb(i')}(y \given x) \lun K^{\bb(i)}(x,y) + 2 K^{\bb(i)}(m,l) \,,
\ees
where $i' = i + O(1)$, $K^{\bb(i)}(x,y)=m$ and $l \leq m$ is to be determined.
\end{propA}

\begin{proof}
Define
\bes
A^i = \{ \langle x',y' \rangle \st K^{\bb(i)}(x',y') \leq m \} \qquad \et \qquad 
A_x^i = \{ y' \st K^{\bb(i)}(x,y') \leq m \} \,,
\ees
which contain $\langle x,y \rangle$ and $y$, respectively. 
A program for $y$ given $x$ is to enumerate $A_x^i$ and to give its enumeration number~$i^y$.
$A_x^i$ can be enumerated if $x$ and $m$ are known. Note that ${\bb(i)}$ is not required, since the enumeration can be done in a parallel fashion until the~$i^y$-th element has been enumerated. This takes a maximum of
\bes
O\left ( 2^{m+1} (1 + 2 + \dots + {\bb(i)}) \right )= O \left( 2^{m} {\bb(i)}^2 \right) 
\ees
steps of computation, namely, enough for each program of length $\leq m$ to be executed (parallel fashion) for ${\bb(i)}$ steps. 
 The exponential factor in~$O(2^m{\bb(i)}^2)$, which can be furthermore bounded by $\leq O(2^{O(\bb(i))}{\bb(i)}^2)$, seems like bad news (it would be for concrete computations). But when compared to~$\bb(i+O(1))$, it is safely ignored. 
 In fact, for any computable function~$f$,
\bes
\bb(i + O(1)) \geq f(\bb(i)) \,,
\ees
because the program 
of length~$i$ that runs for~$\bb(i)$ and the program 
of length~$O(1)$ that computes the function~$f(\cdot)$, can be merged into a program of length~$i + O(1)$ that runs for~$f(\bb(i))$. 

 Let $l \equiv \lceil \log |A_x^i| \rceil $ be the number of bits of~$i^y$. Self-delimitation of the program for~$y$ given $x$ imposes that $l$ must be known in advance. Hence, the program considered requires $K^{\bb(i)}(m,l)$ bits to compute $m$ and $l$ (in time $\leq {\bb(i)}$), and $l$ bits to give the enumeration number of $y$. Any $O(n)$ execution times are absorbed in~${\bb(i+O(1))}$, so 
\be \label{eq:dif1}
K^{\bb(i+O(1))} (y \given x) \lun l + K^{\bb(i)}(m,l) \,.
\ee

\noi Define now
\bes
B^i = \{ x' \st \log |A_{x'}^i| > l-1 \} \,,
\ees
which contains $x$. If $m$ and $l$ are given, $B^i$ can be enumerated by enumerating $A^i$ (thanks to $m$), and when  for a given~$x'$ the subset $A^i_{x'}$ contains more than $2^{l-1}$ elements, $x'$ is added to $B^i$. A possible program for $x$ is thus given by the enumeration number of $x$ in $B^i$. The enumeration of $A^i$ will be completed in time $\leq \bb(i+O(1))$, after which  $x$ is guaranteed to have appeared in the $B^i$ list. Note that
\bes
|A^i| = \sum_{x'} |A^i_{x'}| \geq \sum_{x'\in B^i} |A^i_{x'}| \geq \sum_{x'\in B^i} 2^{l-1} = |B^i| 2^{l-1}\,.
\ees
Since $|A^i| < 2^{m+1}$,
\bes
\log |B^i| \lun m - l  \,.
\ees
This time, the self-delimitation of the enumeration number of $x$ in the $B^i$ list comes for free, since $m-l$ is computed from $m$ and $l$. All together, this amounts to
\be \label{eq:dif2}
K^{\bb(i+O(1))}(x) \lun m-l +  K^{\bb(i)}(m,l)  \,.
\ee
\noi Recalling that $m=K^{\bb(i)}(x,y)$, summing \eqref{eq:dif1} and \eqref{eq:dif2} together yields what is to be shown.
\end{proof}

The~$X$-graph of $\mathcal{L}_x + \mathcal{L}_{y|x}$ represents the function~$K^{B(i)}(x) +  K^{B(i)}(y \given x)$, so by Remark~\ref{remclose}, Proposition~\ref{propcotefacile} implies that $\mathcal{L}_x + \mathcal{L}_{y|x}$ is in an~$O(1)$-neighbourhood of  $\mathcal L_{x,y}$ and Proposition~\ref{prophardside} implies that $\mathcal L_{x,y}$ is in an $O(\log n)$-neighbourhood of $\mathcal{L}_x + \mathcal{L}_{y|x}$. 
Putting this together, one has
\bes
\mathcal L_{x,y} \elog \mathcal L_x + \mathcal L_{y | x} \,.
\ees

\subsection{Not for Description Profiles: The Antistochastic Counter-Example} \label{secfirstanti}

In the light of the equivalence between unrelativized  description and time profiles, Theorem~\ref{thmsophdepth}, it seems that a chain rule analogous to Eq.~\eqref{eqsickrel} should also hold for description profiles. In fact,
\bes
\mathcal L_x  \elog \Lambda_x  \disand  \mathcal L_{x,y}  \elog \Lambda_{x,y} \,, 
\ees
so $\mathcal L_{x\giv y}  \elog \Lambda_{x \giv y} $ holds if and only if $\Lambda_x + \Lambda_{y|x} \elog \Lambda_{x,y} $ holds. But it turns out that these relations are false in general.

Consider the following counterexample. A string~$z$ is called~\emph{antistochastic} if its description profile contains as few elements as possible. More precisely, if~$\len z = n$ and~$K(z) = k$,~$z$ is~$\varepsilon$-antistochastic if~\mbox{$(k-\varepsilon, n-\varepsilon) \notin \Lambda_z$}. 
``All shapes are possible'', Theorem~\ref{thmallshapes}, implies that there exist $O(\log n)$-antistochastic strings.
Within a logarithmic precision, a profile as such is essentially generated by two points, namely,~$(0 , n)$ and~$(k,k)$.
From the description profile perspective, these generators are witnessed by the models~$\{0,1\}^n$ and $\{z\}$, respectively, while from the time perspective, those points come from the programs~``$\mathtt{Print}$~$z$'' and~$z^*$, respectively. See Figure~\ref{figantistoch}.

Antistochastic strings are quite strange: Every model that singles out properties of~$z$ in a more constraining way than just giving raw bits of~$z$ necessarily has complexity~$\geq k$, and every program that computes~$z$ faster than~$\bb(k)$ is as long as the length of~$z$.
Even more impressive, Milovanov~\cite{milovanov2017some} has shown that antistochastic strings have a remarkable holographic property: If~\emph{any}~$n-k$ bits of~$z$ get erased, yielding for instance
\bes
z' = 00*1*01**10011*00*1*111\dots 0\,,
\ees
where the ``$*$'' symbol represents the erased bits, then the original string can be recovered from the erased one by a logarithmic advice, \emph{i.e.},~\mbox{$K(z \given z') = O(\log n)$}.

Let~$z$ be such an~$O(\log n)$-antistochastic string of length~$n$ and complexity~$k$, with \mbox{$n/2 < k < n$}. Let~$z=xy$, where the pieces~$x$ and~$y$ are chosen in such a way that each of them is insufficient to perform Milovanov's holographic reconstruction,~\emph{i.e.}, $\len x < k$ and $\len y < k$.
Technically, $xy$ does not correspond to a proper encoding of the pair $\langle x, y \rangle$ because it is not uniquely decodable, but $1^{||x||}0|x|xy$ is, where~$||x||$ denotes the length of~$\len x$. This discussion holds to logarithmic precision, so the prefix $1^{||x||}0|x|$ can be disregarded, and the profile $\mathcal L_{\langle x, y \rangle}$ is identified to that of $\mathcal L_{xy}$.


\begin{figure}
	\begin{center}
		\includegraphics[
		trim={2cm 14cm 13cm 4cm},clip, width=8cm]{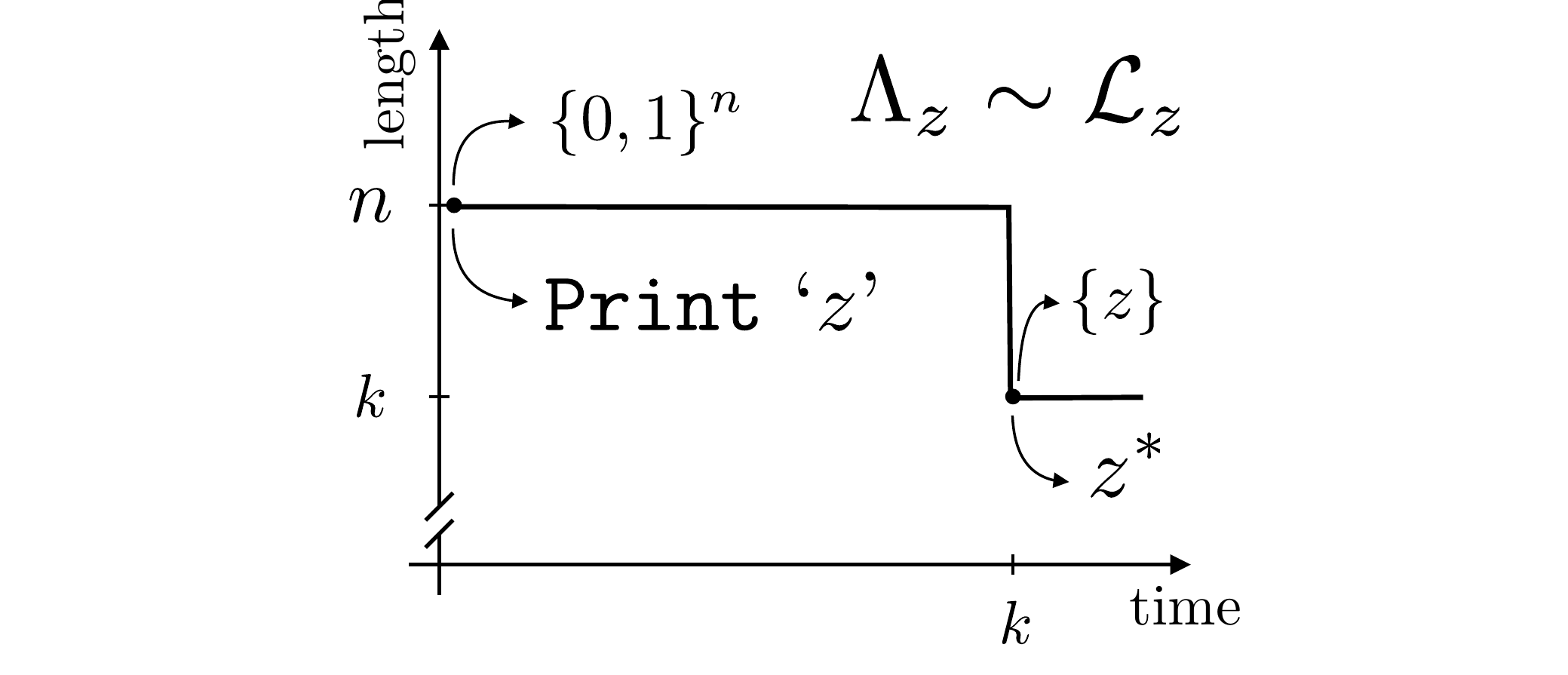}
	\end{center}
\caption{An antistochastic string understood from the description and time profiles perspectives.}
\label{figantistoch}
\end{figure}

Observe that~$x$ is incompressible,~$K(x) \elog \len x$. Otherwise, the set~\mbox{$\{xw \st w \in \{0,1\}^{\len y}\}$} would have complexity smaller than~$\len x$ and with its~$\log$-cardinality of $\len y$, it would entail a two-part description smaller than~$|x|+|y| = n$, contradicting the description profile. This means that, as any incompressible string, $\mathcal L_x$ lies just above the horizontal line of height $\len x$. 

To determine $\Lambda_{y\giv x}$, Milovanov's property implies that
\bes
K(\{y\}\given x) \elog K(y\given x) \elog K(xy\given x) \elog  k - \len x \,,
\ees
and hence~$( k - \len x,  k - \len x) \in \Lambda_{y\giv x}$ (again,  disregarding logarithmic precision). This point happens just after a drop since $(k-|x|-\varepsilon, |y| - \varepsilon) \notin \Lambda_{y\giv x}$, for some $\varepsilon = O(\log n)$. Indeed, if a program of length $k-|x|-\varepsilon$ would, from $x$, specify a model $S\ni y$ of log-cardinality $|y|-\varepsilon - k + |x| + \varepsilon = n-k$, then~$\{xs \st s \in S\} \ni z$ would contradict $z$'s description profile, because it would be of unconditional complexity $k-\varepsilon$, 
 for a two-part description of length $n-\varepsilon$.
\begin{figure}
	\begin{center}
		\includegraphics[
		trim={6cm 9,5cm 18cm 5cm},clip, width=6cm]{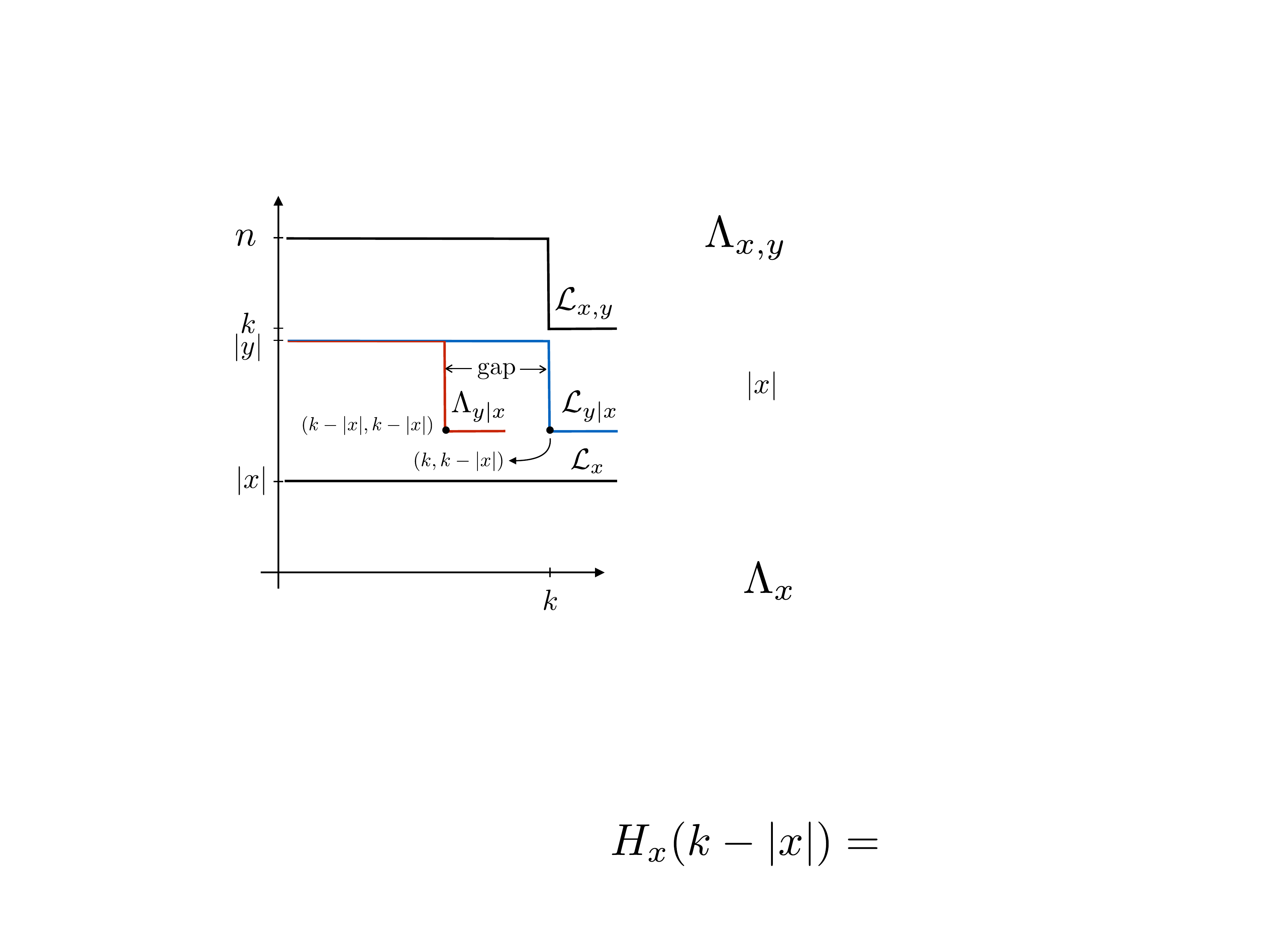}
		\caption{A gap between the conditional profiles~$\mathcal L_{y \giv x}$ and $\Lambda_{y \giv x}$ of pieces~$x$ and~$y$ of anti-stochastic strings shows that the chain rule $\mathcal L_{x,y} \elog \mathcal L_x + \mathcal L_{y \giv x}$ does not find an analogue for description profiles~$\Lambda$.}\label{figdiffprof}
	\end{center} 
\end{figure}

The result of the previous section, Eq.~\eqref{eqsickrel}, suffices to establish $\mathcal L_{y \giv x}$ as $\mathcal L_{x, y} - \mathcal L_{x}$. But for completeness, I argue it directly.
It is straightforward to see that~$(k,  k - \len x) \in \mathcal L_{y\giv x}$: the busy running time of Milovanov's reconstruction cannot exceed the length of the program $(k-\len x)$ plus the length of the auxiliary information $(\len x)$, otherwise, it would solve too big a halting problem from too few bits. Moreover, $(k-\varepsilon, |y| - \varepsilon) \notin \mathcal L_{y\giv x}$, for some $\varepsilon = O(\log n)$ large enough.
Because suppose it is: A program for~$y$ given~$x$ is then of length~$|y| - \varepsilon$ and runs for~$\bb( k - \varepsilon)$ or less steps. A program of length $|x| + |y| - \varepsilon$ for~$z$ is then the~``$\mathtt{Print}$~$x$'', followed by the aforementioned program of~$y$ given~$x$. The overall running time of such program is~$\bb(k - \varepsilon)$, hence contradicting the depth of~$z$: Any program shorter than $n$ that computes $z$ must run for at least~$\bb(k)$. 
%
See Figure~\ref{figdiffprof}.

\smallskip
Answers typically raise more questions: Since~$\Lambda_{y\giv x}$ and~$\mathcal L_{y\giv x}$ do not coincide, what is the gap between them? As a first indicator coming from the previous example, one notices that the first coordinate of the conditional description profile is the conditional complexity (\emph{e.g.}~of value $k-|x|$) of the model, so the length of a program. However, the first coordinate of the conditional time profile is the busy running time (\emph{e.g.}~of value $k$) of a program, which~\emph{could be longer than its length} when auxiliary information is provided.

\section{The Gap} \label{secgap}

On the journey towards expressing sophistication and depth of pairs,
a detour is required to understand --- and quantify --- what separates $\Lambda_{y\giv x}$ from~$\mathcal L_{y\giv x}$.
%
A first step is to understand why the profiles connect in the non conditional case, and then underline what introduces a gap when relativized. To do so, I revisit the non conditional case by introducing another profile ``between'' $\mathcal L$ and $\Lambda$, which renders their link ``more continuous'', thus enlightening why they equate. This new profile is then relativized, allowing us to grasp what causes the gap.

\subsection{A Man in the Middle}
Levin~\cite{levinprivate, vv2002} noticed long ago that strings with description profiles that reach~$K(x)$ for large values of complexity threshold must contain mutual information with the halting problem. This is clear when such a profile is understood by its equivalent time profile, which displays programs that run for so long~($\bb(i)$ steps!)\ that they can decide the halting problem for all programs shorter than~$i$. 

The following profile makes the connection with halting information even clearer. The main idea underlying its construction finds its roots in the proof by G\'acs~\cite{gacs1974symmetry} that some strings~$x$ have high~$K(K(x)\given x)$. The concept is further investigated by Bauwens~\cite{bauwens2010computability} and named \emph{$m$-sophistication}\footnote{
Bauwens was well aware of the connection between sophistication and depth. In fact, in an earlier preprint~\cite{bauwens2009equivalence}, he named the concept \emph{$m$-depth}. This itself is a nice wink to the ``in between profiles'' that is being considered here. 
 }. For the aware reader, here, I restrict the universal semi-measure ``$m$'' to the \emph{a priori} probability, and I put in evidence that the quantity is a function of a significance parameter, hence defining a full-fledged profile.

 To define the \emph{$\chaitin_x$-profile}, consider a dovetailed enumeration of all programs,
in which each program of length $j$, lexicographically, is simulated during $j$ steps of computation, for increasing values of~$j$. Each such iteration refers to a ``\emph{j-step}''. When a program of length~$l$ halts, $2^{-l}$ is added to a sum $M$ initially valued at $0$. Note that this process lower semi-computes~$\Omega$, so as the enumeration goes, an increasing prefix of~$M$ stabilizes to some prefix of~$\Omega$. Whenever $x$ is produced by some program~$p$, the current $j$-step is completed, and a 
dot is marked at the coordinate $(\theta(p), |p|)$, where $\theta(p)$ is the length of the largest prefix of $\Omega$ that has stabilized in the sum $M$, \emph{i.e.}, 
\bes
M_{[\theta]} = \Omega_{[\theta]} \qquad \text{but} \qquad M_{[\theta+1]} \neq \Omega_{[\theta + 1]} \,.
\ees
I refer to~$\theta(p)$, as the \emph{time on the~$\Omega$ clock} and the $\chaitin_x$-profile is defined
as the upwards and rightwards closure of the dots. 
%

Note that~$\theta(p)$ achieves the same purpose as the busy beaver renormalization~$\tau(p)$, that is, it measures the running time of~$p$ in a economical representation. In fact, the two quantities are very close to one another. Indeed,
the busy beaver is friends with a badger, who is also very busy. 
\begin{deffA}
The \emph{\namecbb~function} $\cbb(i)$ is defined as the value of the~$j$ index in the dovetailed enumeration when the first~$i$ bits of~$M$ get stabilized to~$\finite{\Omega}{i}$. This can be written as
\bes
\cbb(i) \deq \min \{j \st M(j) = \finite{\Omega}{i} \beta\} \qquad \text{(for some $\beta$)}\,,
\ees
where~$M(j)$ denotes the value of the sum just before incrementing~$j$ to~$j+1$.
\end{deffA}

 The beaver and the badger can be shown to be almost as busy as one another,
 precisely, that
\be \label{eqcousins}
\bb(i)\leq \cbb(i) \leq \bb(i + K(i) + O(1)) \,.
\ee
To see the first relation, let~$p$ be the slowest halting $i$-bit program, witnessing~$\bb(i)$. In the dovetailed enumeration, when~$p$ halts, the counter~$j$ has value~$\bb(i)$. Before it is supplemented by the contribution~$2^{-i}$, the sum~$M$ cannot have stabilized as large a prefix as $\finom{i}$, otherwise, $M+2^{-i}$ would overshoot the value of $\Omega$.
%
%
The second relation comes from that a program hardcoded with~$\finom i$ can execute the dovetailed enumeration and purposefully halt once $i$ bits of $\Omega$ have stabilized. Such a program has a running time larger than~$\cbb(i)$, but smaller than $\bb(i + K(i) + O(1))$, because it is of length~$\lun K(\finom{i}) \lun i + K(i)$. 
%

By definition, the time of a program~$p$ measured on the~$\Omega$ clock, is the inverse busy badger of its running time, \emph{i.e.}, $\theta(p) = \cbbinv(\rt(p))$. Recalling that $\brt(p) = \bbinv(\rt(p))$, inverting the relation~\eqref{eqcousins} yields
\be \label{eqthetatau}
\tau(p) \llog \theta(p) \leq \tau(p) \,.
\ee

This connection between $\tau$ and $\theta$ establishes the first relation of the following statement, which is in its whole a corollary of the upcoming Propositions~\ref{propcor1} and~\ref{thmlslamcond}. It states that the $\chaitin$-profile can indeed be considered as \emph{a man in the middle} between the $\mathcal L$-profile and the~$\Lambda$-profile.
\begin{corA} \label{thmlslam}
For all~$x$,
\be \label{eqlslam}
\mathcal L_x \elog \chaitin_x \elog 
 \Lambda_x \,.
\ee
\end{corA}

\subsubsection*{A Computable Shape}
The following proposition states that the shape of the $\chaitin_x$-profile can be precisely computed from $x^*$ and its time on the~$\Omega$ clock.

\begin{propA}\label{propcomputingprofile}
For all~$x$,
\begin{center} 
\begin{picture}(120,30)(-45,0)
\put(0,0){\framebox(30,30){$O(1)$}} 
\put(-23,23){\vector(1,0){20}}
\put(-23,7){\vector(1,0){20}}
\put(-43,13){\makebox(20,20){$x^*$}} 
\put(-50,-3){\makebox(20,20){$\theta(x^*)$}} 
\put(33,15){\vector(1,0){20}}
\put(63,5){\makebox(20,20){$\gen(\chaitin_x)$}} 
\end{picture}~~~~~~.
\end{center}
\end{propA}
\begin{proof}
From~$x^*$ and~$\theta$, one can compute~$\finite{\Omega}{\theta}$ by executing the dovetailed enumeration until~$x^*$ halts in the enumeration. Thanks to~$\theta$, one then knows what is the precise prefix of~$M$ that has been stabilized,
 $\finite{\Omega}{\theta}$. With this at hand, one then starts again the dovetailed enumeration, this time, marking a dot at the coordinate~$(i',l)$ when a program of length~$l$ has computed~$x$ in time~$i'$ on the~$\Omega$ clock. One can then return any finite representation of~$\chaitin_x$, for instance, the minimal one $\gen (\chaitin_x)$.
\end{proof}

The previous proposition makes more specific a result by Vereshchagin and Vit\'anyi \cite[\S 7]{vv2002}, which informally states that from $x$, $K(x)$ and the complexity of a near minimal sufficient statistics, a curve $\lambda'$ can be computed, whose closure is logarithmically close to $\Lambda_x$. By the above, this $\lambda'$ can be taken to be~$\partial \chaitin_x$. 


\subsection{Relativizing the $\chaitin$-Profile} \label{sec2profiles}

%
Halting information, too, can be relativized to some auxiliary information~$z$, since it is in general uncomputable to determine whether a program~$p$ yields a halting~$\U(p,z)$. This \emph{relative halting information} 
can again take the form of a halting probability, 
\bes
\Omega^{z} \deq \sum_{p \st \U(p,z) \halts} 2^{-\len{p}} \,.
\ees
By a similar argument as in the unconditional case (\emph{cf.}~Section~\ref{sechalting}),~$\finom{i}^z$ solves the halting problem relative to~$z$, for all programs of length~$\leq i$. Thus, $i \lun K(\finite{\Omega^z}{i} \given z) \lun i + K(i)$.

\subsubsection*{Two ways!} 

It turns out that the $\chaitin$-profile can be relativized in two natural ways. 
To define these conditional profiles, think of a dove with two tails. One dovetailed enumeration runs all programs of length $\leq j$, lexicographically, for $j$ steps of computation on the reference computer~$\U$,~\emph{supplemented by auxiliary information}~$z$. If~$\U(p,z) \halts$ for some program~$p$, then~$2^{-\len p}$ is added to a sum~$M^z$. 
Before incrementing the~$j$ counter, the other tail is visited, running the same programs, also for~$j$ steps of computation on~$\U$, but \emph{without auxiliary information}. If~$\U(q) \halts$ for some program~$q$, then~$2^{-\len q}$ is added to a different sum~$M$.  I shall refer to~$M^z(j)$ and~$M(j)$ as the value taken by the sums just before incrementing the counter to~$j+1$. The idea is to have two clocks to measure time: one follows the stabilization of a prefix of~$\Omega^z$ by~$M^z$ and the the other, of~$\Omega$ by~$M$.

Whenever the first enumeration finds a~$p$ such that~$\U(p,z)=y$, the $j$\textsuperscript{th} step is completed and a \textcolor{red}{red} dot is marked at the coordinate $(\theta^z(p), |p|)$, where $\theta^z(p)$ is the length of the largest prefix of $\Omega^z$ that have stabilized in $M^z(j)$. I shall call~$\theta^z(p)$ the~\emph{time on the $\Omega^z$ clock}. Additionally,
 a \textcolor{blue}{blue} dot is marked at the coordinate $(\theta(p), |p|)$, where $\theta(p)$ is as before, the time on the $\Omega$ clock.
Define~\textcolor{red}{$\chaitina_{y \giv z}$} and~\textcolor{blue}{$\chaitinb_{y\giv z}$} 
 as the upwards and rightwards closure of the \textcolor{red}{red} and \textcolor{blue}{blue} dots, respectively. 
 
\begin{remA} \label{remallthesame}
$\chaitin^\varepsilon_{y\giv \varepsilon} = \chaitinb_{y \giv \varepsilon} = \chaitin_{y}$.
\end{remA}

\begin{propA} \label{propcalcul}
Let $\sigma$ be the time \emph{either} on the $\Omega$ or on the $\Omega^z$ clock when $(y \giv z)^*$ halts, then

\smallskip
\begin{center} 
\begin{picture}(120,46)(-45,0)
\put(0,0){\framebox(30,46){$O(1)$}} 
\put(-23,39){\vector(1,0){20}}
\put(-23,23){\vector(1,0){20}}
\put(-23,7){\vector(1,0){20}}
\put(-43,29){\makebox(20,20){$z$}} 
\put(-48,13){\makebox(20,20){$(y\giv z)^*$}} 
\put(-43,-3){\makebox(20,20){$\sigma$}} 
\put(33,32){\vector(1,0){20}}
\put(63,22){\makebox(20,20){$\gen(\chaitina)$}} 
\put(33,14){\vector(1,0){20}}
\put(63,4){\makebox(20,20){$\gen(\chaitinb)$}} 
\end{picture}\,~~~.
\end{center}
\end{propA}
\begin{proof}
The proof is analogous to that of Proposition~\ref{propcomputingprofile}. 
\end{proof}

\subsubsection*{Each with his Own Mate}
The following two propositions state that each version of the relative $\chaitin$-profiles follows its own other relative profile: 
\bes
\chaitina_{y\giv z} \elog \Lambda_{y \giv z}\qquad \text{while} \qquad \chaitinb_{y\giv z} \elog \mathcal L_{y \giv z}\,.
\ees

\begin{propA} \label{propcor1}
For all $y$ and $z$,
\bes
\mathcal L_{y\giv z} \subseteq \chaitinb_{y\giv z} \disand
\chaitinb_{y\giv z}  \subseteq O(\log n)\text{-neighbourhood of } \mathcal L_{y\giv z} \,.
\ees
\end{propA}
\noi This  implies that $\mathcal L_{y\giv z} \elog \chaitinb_{y\giv z}$. 
\begin{proof}
A program~$p$ is the fastest witness of a point in~$\gen (\mathcal L_{y \giv z})$ if and only if it is the fastest witness of a point in $\gen(\chaitin_{y \giv z})$, both points being horizontally aligned at height $\len p$. By taking such a fastest witness for $(\tau (p), |p|) \in \gen (\mathcal L_{y \giv z})$ and $(\theta(p), \len p) \in \gen(\chaitin_{y \giv z})$, the conclusion follows from Equation~\eqref{eqthetatau} and remark~\ref{remclose} (i).
\end{proof}
\begin{propA} \label{thmlslamcond}
For all $y$ and $z$,
\beas
\Lambda_{y\giv z} &\subseteq & O(1)\text{-neighbourhood of }\chaitina_{y\giv z} \disand \\
\chaitina_{y\giv z} &\subseteq & O(\log n)\text{-neighbourhood of }\Lambda_{y \giv z} \,.
\eeas
\end{propA}
\noi This means that $\Lambda_{y\giv z} \elog \chaitina_{y\giv z}$.

\begin{proof}
\nit{$\Lambda_{y\giv z} \subseteq  O(1)\text{-neighbourhood of }\chaitina_{y\giv z}$.}

\noi Let~$S\ni y$ be a model witnessing $(i,\lambda) \in \gen(\Lambda_{y\giv z})$. It induces a program that computes~$y$ from~$z$ via its two-part description, of length~$\lambda$. But what is its time on the $\Omega^z$ clock? Being~$i$-bit long, the first part runs for at most time~$i+ O(1)$ on the~$\Omega^z$ clock, which is the most conservative bound for an $i$-bit program running with auxiliary information $z$. 
And \emph{the second part of a two-part description is fast}: It takes~$O(|S|) \leq O(|\{0,1\}^{\len y}|)$ steps, which is negligible compared to the time bound of the first part, so it can be absorbed by increasing the time on the~$\Omega^z$ clock to $i+O(1)$.  

\nit{$\chaitina_{y\giv z} \subseteq  O(\log n)\text{-neighbourhood of }\Lambda_{y \giv z}$.}

\noi
Let $(i, \ell)\in \gen(\chaitina_{y\giv z})$ be witnessed by a program~$p$ computing~$y$ given~$z$, of length~$\ell$ and time $i$ on the $\Omega^z$ clock. Programs can be grouped together on the basis of their length and their time on the $\Omega^z$ clock. 
Hence, for arbitrary $l$, define 
\beas
\tilde A_{i,l} &=& \{r \st |r| = l \et \theta^z(r)\geq i\} \,, \\
\bar A_{i,l} &=& \{r \st |r| = l \et \theta^z(r)=i\}  ~ \text{and} \\
A_{i,l} &=& \{\U(r) \st |r| = l \et \theta^z(r)=i\} \,.
\eeas
%
%
Notice that $p$ is an element of the first two sets and that~$y$ is an element of the~$A_{i,\ell}$. I shall show that~$A_{i,\ell}$ is a model with
\bes
K(A_{i,\ell}\given z) \llog i \disand K(A_{i,\ell}\given z) + |A_{i,\ell}| \llog \ell \,.
\ees 

First, observe that given $z$, $ A_{i,\ell}$ can be computed from $\bar A_{i,\ell}$, which can be computed from $\finite{\Omega}{i}^z$ and~$\ell $, so
\bes
K( A_{i,\ell} \given z)  \lun  K(\finite{\Omega}{i}^z, \ell \given z) \llog i \,.
\ees
Second, the log-cardinality of~$A_{i,\ell}$ needs to be bounded, and because it contains fewer elements than $\tilde A_{i,l}$, bounding the latter suffices.  Define~\mbox{$a_{il} \equiv |\tilde A_{i,l}|$}. For a fixed~$i$, the discrete application $l \mapsto a_{il}$ is lower semi-computable from~$z$ and~$\finite{\Omega}{i}^z$. Moreover,
\bes
\sum_l a_{il} 2^{-l} \leq 2^{-i} \,,
\ees
otherwise, too large of an algorithmic mass of programs would remain to halt --- contradicting the~$i$-th bit of~$\Omega$. This means that~$a_{il} 2^{-l+i}$ is a lower semi-computable semi measure, relative to~$z$ and~$\finite{\Omega}{i}^z$, so by the coding theorem\footnote{
The coding theorem~\cite{Li2008} states that every discrete application $j \mapsto \mu(j)$ that is 
 (i) lower semi-computable from auxiliary information~$z$ and
 (ii) a semi-measure, \emph{i.e.}, $\sum_j \mu(j) \leq 1$,
 has~$\mu(j) \leq 2^{-K(j\given z) + O(1)}$.},
\bes
a_{il} \leq 2^{l-i-K(l \given z, \finite{\Omega}{i}^z, K(\finite{\Omega}{i}^z \giv z))+O(1)} \,.
\ees
Therefore,~$\log |A_{i,\ell}| \leq \log a_{i\ell} \lun  \ell-i-K(\ell \given z, \finite{\Omega}{i}^z, K(\finite{\Omega}{i}^z \giv z))$. So
\bea \label{eqsuffstatcond}
K(A_{i,\ell})+ \log |A_{i,\ell}| &\lun& K(\finite{\Omega}{i}^z, \ell \given z) + \ell-i-K(\ell \given z, \finite{\Omega}{i}^z, K(\finite{\Omega}{i}^z \giv z))\nonumber \\
&\eun& K(\finite{\Omega}{i}^z \given z) + \ell - i \nonumber \\
&\lun&   \ell + K(i)  \\
&\llog & \ell \,.\nonumber
\eea
\end{proof}

Corollary~\ref{thmlslam} thus follows from the last two propositions and from Remark~\ref{remallthesame}. This corollary is the equivalence between depth and sophistication.
 Minimal sufficient statistics induce a two part-code in a way that forces triviality, and hence fast computation, of the second part.  This then distillates all the slow computation  (\emph{i.e.}, the deep structures in Bennett's sense) into the model (the sophisticated structures in Kolmogorov's sense). In an algorithmic information theoretic sense, those deep and sophisticated structures essentially made of initial segments of $\Omega$; they are full of halting information. 
 
%

\subsection{Losing Synchronicity}

Light can now be shed on the difference between the conditional profiles~$\mathcal L_{y\giv z}$ and~$\Lambda_{y\giv z}$, through their equivalent representations in terms of conditional~$\chaitin$-profiles. 
The difference between $\chaitina_{y \giv z}$ and $\chaitinb_{y \giv z}$ is only a horizontal distortion, since generators come in horizontally  aligned pairs as they are witnessed by the same program~$p$ whose length establishes the second coordinate. The distortion reflects that of the clocks $\Omega^z$ and~$\Omega$, with respect to which the times $\theta^z(p,z)$ and $\theta(p,z)$ determine the first coordinate. 

As a first step to better characterize the difference in the flow of the clocks, the program~$p$ witnessing the aforementioned aligned generators can be abstracted, and rely only on the times $\theta^z$ and $\theta$ showed by the clocks.

\begin{deffA}
Define the \emph{relativized \namecbb}, $\condcbb{z}(i)$, as the value of the $j$ index when $i$ bits of $\Omega^z$ have stabilized, namely,
\bes
\condcbb{z}(i) \deq \min \{j \st M^z(j) = \finite{\Omega^z}{i} \beta\} \,.
\ees
\end{deffA}
\noi Observe that
\bes
\condcbb{z}(\theta^z) \leq  \rt(p,z)  < \condcbb{z}(\theta^z + 1) \implies \cbbinv \condcbb{z} (\theta^z) \leq \theta \leq \cbbinv \condcbb{z} (\theta^z + 1) \,,
\ees
so the connection from the $\Omega^z$ to the $\Omega$ clock is $\cbbinv \condcbb{z} (\cdot)$.
In what follows, this connection is reframed in terms of wether --- and if so \emph{how} --- $z$ has information about the halting problem.


\subsubsection*{Relation with Halting Knowledge}

First, I exemplify this connection. Suppose~$z=\finom{a}$, and $\theta^z=b$, what is the corresponding time~$\theta$ on the~$\Omega$ clock? If $b$ bits of~$\Omega^z$ have stabilized, it means that no more programs shorter than~$b$ bits in length will ever lead to a halting~$\U(\cdot,z)$ computation, in particular, the program described in the following paragraph.

 
  With the help of $z$ and some extra hardcoded bits of~$\Omega$, assemble~$\finom{a+b-O(\log b)}$, and execute the dovetailed enumeration of programs, run without~$z$, until the sum exceeds~$\finom{a+b-O(\log b)}$. This particular computation takes a time $ a+b - O(\log b)$ on the~$\Omega$ clock, so~$\theta \glog a+ b$. By incompressibility of~$\Omega$, in fact~$\theta \elog a+ b$ holds. This example puts in evidence that the gap between $\theta^z$ and~$\theta$ depends upon~$z$'s knowledge about the halting problem. More precisely, the distortion in the flow of the clocks turns out to be a property of \emph{the manner} in which $z$ has such knowledge.

In the spirit of the above example, the following definition quantifies how close to~$\Omega$ one can get from $z$ and~$i$ bits of advice.
\begin{deffA}
The \emph{reach curve} of $z$ is defined as 
\bes
\spit_z(i) \deq \max \{ s \st K(\finite{\Omega}{s} \beta \given z) \leq i \et \finite{\Omega}{s} \beta < \Omega \} \,.
\ees
\end{deffA}
This definition is reminiscent of monotone complexity, where the finite string $\beta$ is a tool for an overall $\finite{\Omega}{s} \beta$ possibly simpler than the raw $\finite{\Omega}{s}$, generally for length reasons. 
The terminology has a twofold interpretation. $\spit_z(i)$ measures how close to~$\Omega$ can be reached, which is directly connected to how large a number (or running time) can be reached.
If~$z$ is independent from the halting problem, its reach curve follows the identity line within logarithmic resolution: $i$ bits of program grants $\elog i$ bits of prefix of $\Omega$. However, if~$z$ contains pieces of information about~$\Omega$ its reach curve will display the benefits of that knowledge by moving above of the identity line. 

The following proposition pinpoints \emph{what} information is the most helpful for~$z$ to reach as large a prefix of~$\Omega$ as possible. In other words, what should the~$i$ bits of advice be made of? The answer is the initial bits of~$\Omega^z$. Hence, if~$z$ has holes in its halting knowledge, then~$\Omega^{z}$ fills them.

\begin{propA}
Let $\spit_z(i)=r$ be witnessed by the program $p$ such that $\U(p,z)=\finom{r}\beta < \Omega$ and~$\len p \leq i$. Then~$p$'s algorithmic information is essentially that of~$\finom{i}^z$, since
\bes
K(\finom{r} \beta \given z, \finom{i}^z) = O(1) \,.
\ees
\end{propA}
\begin{proof}
From $z$ and $\finom{i}^z$, one can compute the list~$\U(q,z)$, for all halting programs~$q$ of length~$\leq i$ (the non-halting programs are discarded). Each such~$q$ can be transformed in a program~$O(1)$ longer that I shall call the the $\U(q,z)$-dovetail. This consists of the dovetailed enumeration of all programs, run without~$z$, until the sum $M$ exceeds the value of~$\U(q,z)$ previously computed. The halting status of each $\U(q,z)$-dovetail can be obtained from $z$ and $\finom{i+O(1)}^z$. The latter is non-constructively acquired by the~$O(1)$ advice. The largest $\U(q,z)$ leading to a halting $\U(q,z)$-dovetail is then outputted.
\end{proof}

The next proposition states that the reach curve~$\spit_z(i)$ expresses equivalently the connection~$ \cbbinv \condcbb{z}(i)$ between clocks: they are logarithmically close to one another. Both relations are non-decreasing, so their upwards and \emph{leftwards} closure define the respective profiles $\mathcal R_z$ and $\mathcal B^{-1} \mathcal B_z$. Since the profiles can go beyond the length of~$z$, the upcoming~``$\elog$'' relation refers to $O(\log i)$, where~$i$ is the first coordinate of the profiles' points.

\begin{propA}\label{proprbb}
For all~$z$, 
\bes
\mathcal R_z \elog \mathcal B^{-1} \mathcal B_z \,.
\ees
\end{propA}
\begin{proof}
It suffices to show that~$\spit_z(i') \geq \cbbinv \condcbb{z}(i)$, for $i' \llog i$ and vice versa. 
Let~$\cbbinv \condcbb{z}(i)=r$, so when the double dovetailed enumeration is performed, when~$i$ bits of~$\Omega^z$ stabilize,~$r$ bits of~$\Omega$ are stabilized. A program of length~$\llog i$, with knowledge of $\finom{i}^z$, can then compute $\finom{r}\beta$ by also running the two dovetailed enumerations, and when $\finom{i}^z$ is stabilized in one enumeration, outputs the sum $M=\finom{r}\beta$ of the other%
\footnote{Can it be shown that the monotone complexity of $\finom{i}$ is smaller than $i + O(1)$, \emph{i.e.}, $\forall i \exists \gamma K(\finom{i}\gamma) \lun i$? If so the 2 profiles would be $O(1)$ close.}.

Now, I show that~$\cbbinv \condcbb{z}(i')\geq \spit_z(i)$, for $i' \lun i$. Let $\spit_z(i) = r$, so $\finom{r}\beta = \U(p,z)$ for some $p$ of length $\leq i$. This program can be transformed into the $\finom{r}\beta$-dovetailing, of length~$i' =i + O(1)$. Therefore $\condcbb{z}(i')$ is no smaller than the running time of the~$\finom{r}\beta$-dovetailing, which is long enough to stabilize~$r$ bits on the~$\Omega$ clock, so $\condcbb{z}(i') \geq \cbb(r)$.
\end{proof}

\begin{figure}
	\begin{center}
		\includegraphics[
		trim={4cm 10cm 18cm 5cm},clip, width=8cm]{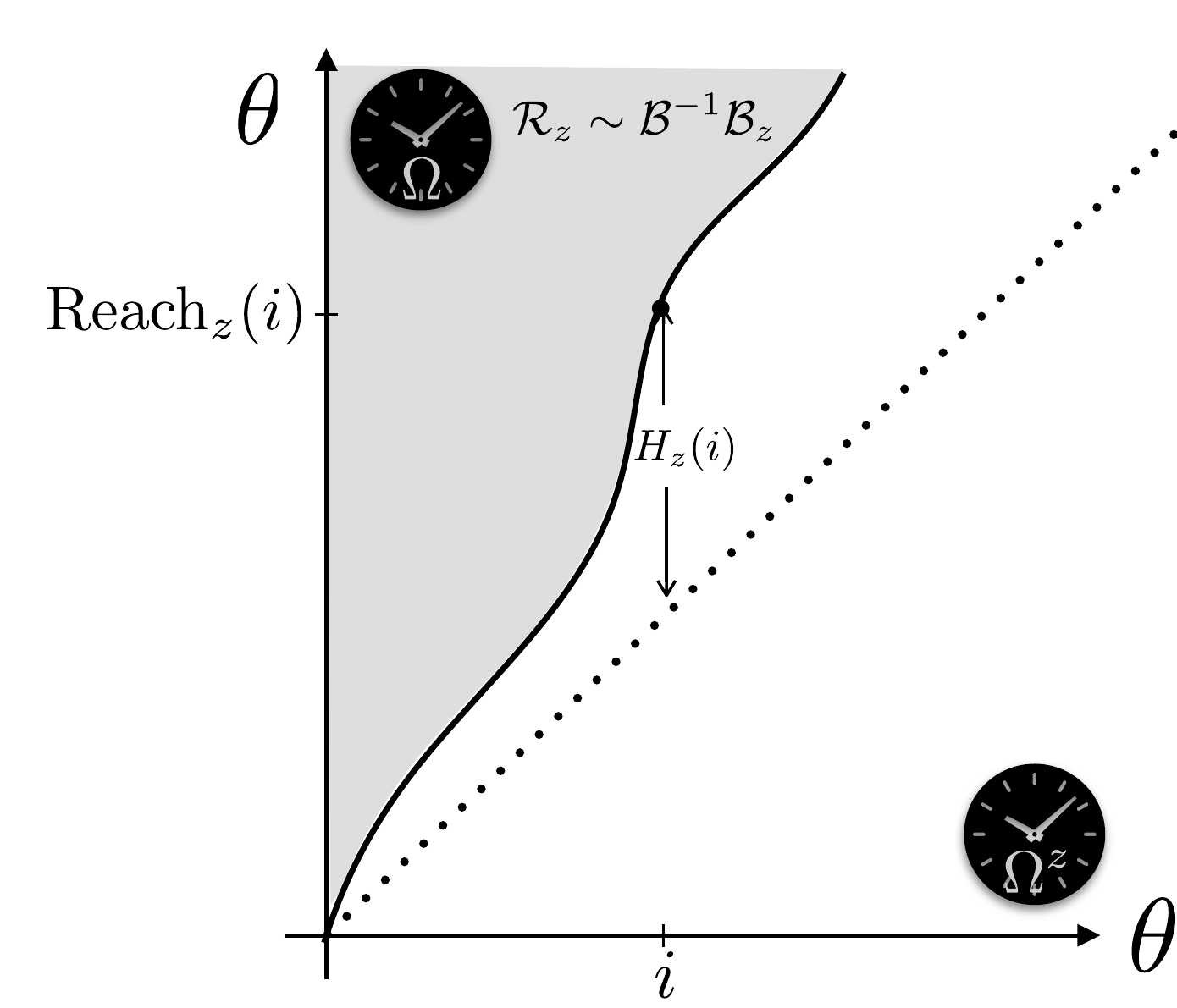}
	\end{center}
\caption{The connection between the $\Omega$ and the $\Omega^z$ clock is given by the reach curve of $z$.}
\label{connection}
\end{figure}

\subsubsection*{Naming the Gap}

\begin{figure}
	\begin{center}
		\includegraphics[
		trim={5cm 8cm 14cm 7cm},clip, width=8cm]{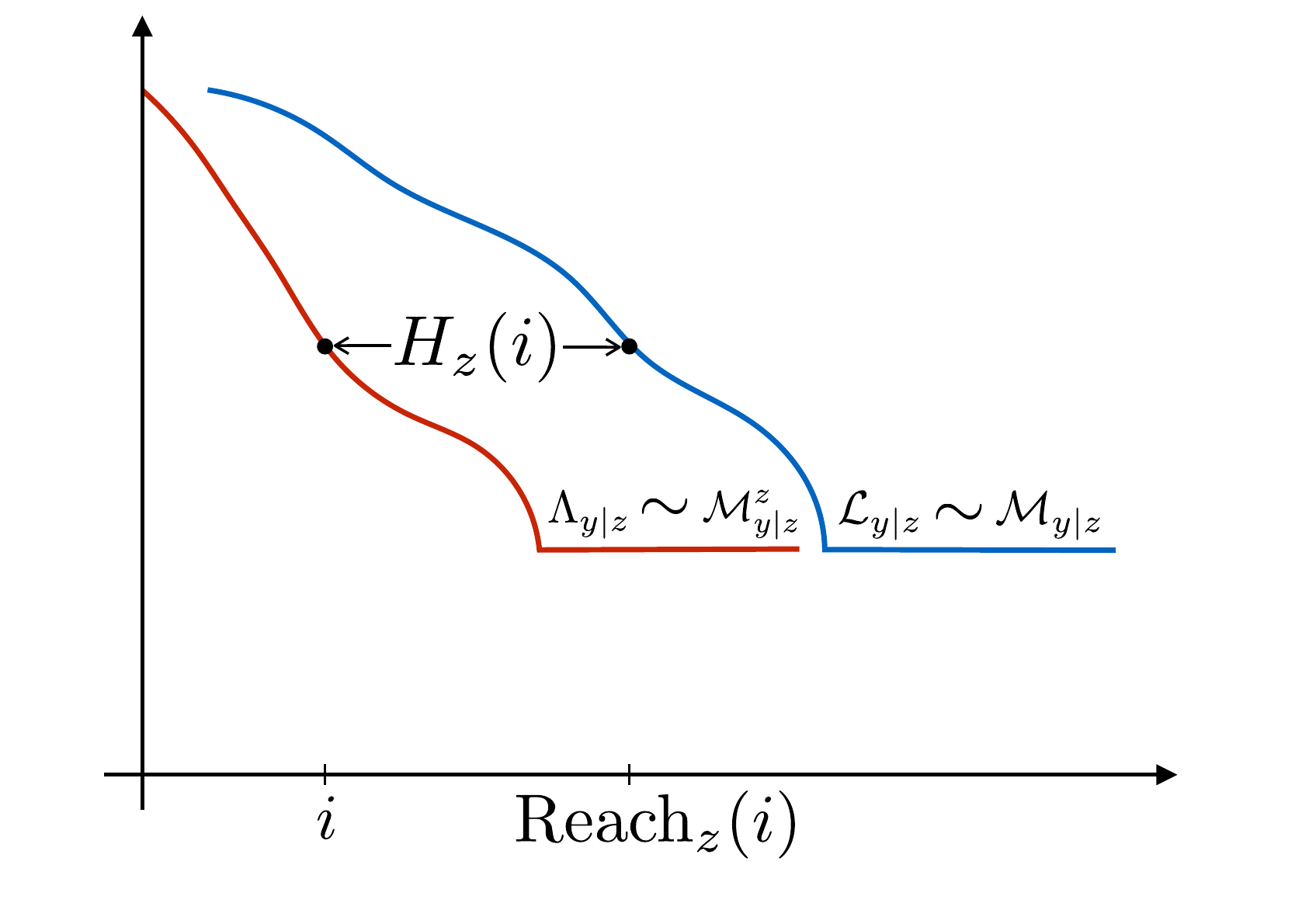}
	\end{center}
\caption{The conditional profiles and the gap between them.}
\label{Hz}
\end{figure}

Of interest is the quantity
\bes
\hd_z(i) \deq \spit_z(i) - i \,,
\ees
which 
measures the time \emph{difference} between the $\Omega$ and the $\Omega^z$ clocks, hence, the gap between the relative profiles.
See Figures~\ref{connection} and~\ref{Hz}.
 Being an affine transformation of~$\spit_z(i)$, it encodes the same information. 
 
$\hd_z$ may be called the \emph{halting materialization distribution} because of the following observations.
For small values (logarithmic in the length of~$z$), the halting materialization distribution coincides with the reach curve,
\bes
\hd_z(O(\log n)) \elog \spit_z(O(\log n)) \,,
\ees
and represents the largest prefix of~$\Omega$ that can be computed from a logarithmic advice (such halting information materializes easily). This value is an important characteristic of strings with any sort of interesting profiles. In fact, a string~$z$ that displays a drop at value~$d$ in his time profile~$\mathcal L_z$, will have~$\hd_z(O(\log n)) \glog d$. This is because such a drop witnesses that the fastest program of a certain length $\ell$, that computes~$z$, runs for~$\bb(d)$ steps of computation, long enough to stabilize almost~$d$ bits of~$\Omega$. Therefore, with~$z$ at hand, an~$O(\log n)$ advice to reach close to~$\Omega$ is simply~``$\ell$''. It serves as a promise of finding an~$\ell$-bit long program that computes~$z$. In the process of finding it, the sum~$M$ of the dovetailed enumeration will stabilize $\elog d$ bits of~$\Omega$.

And at the other end of the spectrum,~$\lim_{i \to \infty} \hd_z(i) = I(z \colon \Omega)$. In fact,
\beas
\hd_z(i) &=& \max \{s-i \st K(\finom{s}\beta \given z ) \leq i \} \\
&\elog &  \max \{s- K(\finom{s}\beta \given z ) \st K(\finom{s}\beta \given z ) \leq i \} \\
&\elog &  \max \{K(\finom{s})- K(\finom{s} \given z ) \st K(\finom{s}\beta \given z ) \leq i \} \\
&\elog &  \max \{I(\finom{s} \st z ) \st K(\finom{s}\beta \given z ) \leq i \} \,.
\eeas
As~$i$ grows,~$s$ grows at the same pace or faster, so in the limit~\mbox{$i \to \infty$},~$s$ goes also to~$\infty$.
 In between small and large values, the shape of~$\hd_z$ informs us of how hard it is to materialize the halting knowledge of~$z$. For instance, $\zeta$ made of bits number 501 to 2000 of $\Omega$ is useless to solve \emph{any} halting problem... until a clever 500-bit advice is provided. In such a case, the halting information of~$\zeta$ is only \emph{materialized} after $i=500$, and $\hd_\zeta$ is indeed a step function, with the step at that value.

 If the antistochastic strings looked like the strangest of all in the view of their $\Lambda$ and their $\mathcal L$ profiles, still, they display a relatively straightforward halting materialization distribution: It is constant at the value corresponding to the drop of the $\mathcal L$ profile, which is at value of their complexity. 
More elaborate halting materialization profiles are possible and in fact, the following proposition shows that all shapes are possible.
 
\begin{propA}
For any non-decreasing function~$h(i)$ that eventually remains constant, there exist a string~$\gamma$ whose halting materialization distribution~$\hd_\gamma(i)$ is~$O(K(h))$ close to~$h(i)$, where $K(h)$ is the complexity of the function $h$, which is defined as $\min_p\{|p| \st p \text{ computes }h \}$.
\end{propA}
\begin{proof}
This proof is about playing a game with the bits of~$\Omega$, in which one basically encodes the graph of~$h(i)$ into $\gamma$, as to which bits of $\Omega$ are given. Let $a_0$ be the first integer mapped to a non-null value, $b_0 = h(a_0)$. Let $a_1$, $a_2$, \ldots, $a_m$ all the values at which~$h$ increases, and~$b_1$,~$b_2$, \ldots,~$b_m$ the corresponding amounts by which $h$ increases. Define
$$
\gamma \equiv \Omega_{a_0} \dots \Omega_{a_0+b_0} \Omega_{a_0+b_0+a_1} \dots \Omega_{a_0+b_0+a_1+b_1} \Omega_{a_0+b_0+a_1+b_1+a_2} \dots
\Omega_{\sum_0^m a_i \sum_0^m b_i} \,,
$$
where~$\Omega_c$ stands for the~$c$-th bit of~$\Omega$.
From~$\gamma$ and~$\eun i + K(h, i)$ bits of advice, a prefix of~$\Omega$ is obtained by ``patching its holes'' with a string of length~$i$ defined as
\bes
\delta = \Omega_{1} \dots \Omega_{a_0-1} \Omega_{a_0+b_0+1} \dots \Omega_{a_0+b_0+a_1-1} \dots \Omega_{h(i) + i -1}\Omega_{h(i) + i} \,.
\ees 
The extra $K(h,i)$ bits are required for delimitation purposes: Not only self-delimitation of~$\delta$, but mostly to unravel the bits of~$\gamma$ and the bits of~$\delta$ in order to assemble~$\finom{i+h(i)}$.
This particular choice of advice shows that
\bes
\spit_\gamma(i + K(h, i) + O(1)) \geq i + h(i) \,.
\ees
A program of such a length could not compute a larger prefix, since it would contradict the incompressibility of~$\Omega$.
\end{proof}



Let me return to where we started. The conditional profiles~$\Lambda_{y\giv z}$ and~$\mathcal L_{y \giv z}$ do not correspond. They have been shown to be equivalently represented by the profiles~$\chaitina_{y \giv z}$ and~$\chaitinb_{y \giv z}$, respectively, whose difference is a horizontal distortion quantified by~$H_z(i)$. 
This distortion measures the difference of flow between the~$\Omega^z$ and the~$\Omega$ clocks, which is related to the difficulty of~$z$ to materialize its halting information in terms of a prefix of $\Omega$.

\section{Depth and Sophistication of Pairs} \label{secpair}

As mentioned in Section~\ref{sec:crtp}, a cornerstone of the algorithmic theory of information is the chain rule, eq.~\eqref{eqcr}, which relates the complexity of a pair 
 to that of a single string and a conditional homologue. Logical depth and sophistication arose from an effort to measure the meaningful information in a string, and not just its randomness. In the light of the previous results, depth and sophistication of pairs can now  
be expressed in terms of their single string and conditional versions.

For tidier expressions characterizing depth and sophistication for pairs, one should free the concepts from their significance parameters, keeping only the essence of what they capture. 
This is achieved when the significance parameters are taken as small as possible. 
%
\subsection{Depth$_0$}

For the busy beaver depth, the natural candidate of a parameter-free version is $\bdepth_0(\cdot)$. 
It amounts to the busy running time of the (fastest) shortest program. The significance parameter of the busy beaver depth can meaningfully be taken to $0$, because time profiles are not naturally bumpy: Even the smallest drop of one unit deep in the $\mathcal L_x$ profile of some string~$x$ is very significant. Such a drop grasps that~$x$ contains a lot of mutual information with a prefix of~$\Omega$, simply through the running time of its shortest program.

However, such a micro drop as the last drop of the profile is problematic in the task of formulating a relation between $\bdepth_0(x,y)$, $\bdepth_0(x)$ and $\bdepth_0(y\given x)$, since the main tool at hand is the relation~\eqref{eqsickrel}, $\mathcal L_{x,y} \elog \mathcal L_{x} + \mathcal L_{y\giv x}$, which incorporates errors of logarithmic order on the $Y$ axis. Recall that the depth profile $\mathcal D$, Eq.~\eqref{eqdepthprofile}, is a downwards translation of the time profile $\mathcal L$, with $\bdepth_c$ being represented as the $\graphy$ of~$\mathcal D$. Consequently, the errors of logarithmic order transpose on the axis of the depth's significance parameter. 
To keep the discussion grounded in the ideas, I will avoid the conundrum by imposing an extra constraint on the considered profiles. 
The strings $x$ and $y$ are said to have $\mathcal L$-profiles with a \emph{sharp finish} if all their time profiles~(\emph{e.g.}, $\mathcal L_{x,y}$, $\mathcal L_{x \giv y}$, \dots) display a last drop that is greater than some $\varepsilon = O(\log n)$. More precisely, the parameter $\varepsilon$ is chosen greater than the sum of the error terms in Propositions~\ref{propcotefacile} and~\ref{prophardside}. This ensures that the latest drop of $\mathcal L_{x,y}$ is aligned (up to $O(1)$ resolution) with either the latest drop of $\mathcal L_{x}$ or with the latest drop of $\mathcal L_{y \giv x}$. 
The $X$ coordinate at which the latest drops happen marks the~$\bdepth_0$. Therefore, if $x$ and $y$ have profiles with a sharp finish,
\be \label{eqdepthpair}
\bdepth_0(x,y) \eun \max \{ \bdepth_0(x), \bdepth_0(y \given x)\} \,.
\ee
This relation means that the running time (in busy beaver units) of the shortest program that produces the pair $x,y$ is close to either that of $x^*$ or that of $(y\giv x)^*$. Since the relation can instead be developed on $y$ and $x\given y$, if~$x$ or~$y$ is deep, so is the pair. However, the reciprocal does not hold. When $x$ and $y$ are pieces of an antistochastic string, each of them is individually shallow but deep relative to one another, yielding a deep pair.

Finally the deterministic \emph{slow growth law} \cite{bennett1995logical} can be retrieved from Equation~\eqref{eqdepthpair}. In fact, let $y$ be a computable processing of~$x$. 
If $x$ is shallow, but $y$ is deep, then the relation implies that $y$ is deep relative to $x$: it cannot have been computed by a short \emph{and} fast program.

\subsection{A Parameter-free Sophistication?}

Exhibiting a parameter-free notion of sophistication is a more sophisticated task ;-). In an aphorism, sophistication is the complexity of the minimal sufficient statistic, but then, what is the precise criterion for a statistic to be sufficient?
A sufficient statistic is often~(\emph{e.g.}, \cite[\S 2]{vv2002} \cite[\S 5]{vitanyi2006meaningful} \cite{gacs2001algorithmic}) defined to be an~$S \ni x$ that satisfies
\be \label{eqtoodiff}
K(S) + \log|S| = K(x) + O(1)\,.
\ee
However, the nature of two-part descriptions generally makes this relation too difficult to satisfy.

Before I elaborate more on this, I must mention that Antunes and Fortnow~\cite{antunes2009sophistication} approached the problem of liberating sophistication from its parameter by including it in the minimization. \emph{Coarse sophistication} is thus defined as
\bes
\text{cSoph}(x) = \min_c~\{\soph_c(x) + c\} \,.
\ees
This definition suffers from the problem that it does not do justice to the most sophisticated strings of a fixed length~$n$. Indeed those have an antistochastic-like profile, with a drop (of height $\delta = n-K(x)$) as late as possible (at $K(x)$). A late drop as such forces $K(x)$ to be close to $n$, thereby shrinking the height $\delta$ of the drop. Consider a string~$x$ as such with $\delta$ small, but still in $\Omega(n)$. Its sophistication is large: $\soph_c(x)=K(x)$, for $c \leq \delta$, as witnessed by its only minimal sufficient statistic $\{x\}$. However, its coarse sophistication collapses to $\delta$, as witnessed by $\{0,1\}^n$.

I come back to the perhaps too strict constraints of the criterion of Eq.~\eqref{eqtoodiff}. As mentioned in the preliminaries, the shortest one-part description for $x$, this is~$x^*$, in itself carries more algorithmic information than $x$ alone: It carries its own length~$K(x)$, 

\smallskip
\begin{center} 
\begin{picture}(120,30)(-45,0)
\put(0,0){\framebox(30,30){$O(1)$}} 
\put(-23,15){\vector(1,0){20}}
\put(-43,5){\makebox(20,20){$x^*$}} 
\put(33,23){\vector(1,0){20}}
\put(33,7){\vector(1,0){20}}
\put(53,13){\makebox(20,20){$x$}} 
\put(58,-3){\makebox(22,20){$K(x)$\,.}} 
\end{picture}
\end{center}

For the same self-delimitation reason, a two-part description $D(S^*, i^x_S) = \alpha S^* i_S^x$ carries in itself two implicit lengths: those of each part. 
 Thereby,
%

\smallskip
\begin{center} 
\begin{picture}(120,46)(-45,0)
\put(0,0){\framebox(30,46){$O(1)$}} 
\put(33,39){\vector(1,0){20}}
\put(33,23){\vector(1,0){20}}
\put(33,7){\vector(1,0){20}}
\put(53,29){\makebox(20,20){$x$}} 
\put(59,13){\makebox(20,20){$K(S)$}} 
\put(62,-3){\makebox(25,20){$\log |S|$ \,,}} 
\put(-23,32){\vector(1,0){20}}
\put(-23,14){\vector(1,0){20}}
\put(-43,22){\makebox(20,20){$S^*$}} 
\put(-43,4){\makebox(20,20){$i^x_S$}} 
\end{picture}
\end{center}
so $K(S) + \log|S| \gun K(x, K(S), \log|S|) \eun K(x) + K(K(S), \log|S| \given x, K(x))$.

For~$K(S) + \log|S|$ in the vicinity of~$K(x)$, the extra complexity brought by the last term is essentially that of a delimiter, $K(S)$, that breaks the number~$K(x)$ in two pieces. Arguments can be made that by increasing the value of that delimiter, it will eventually be of small complexity, given~$K(x)$. But can this ``small'' be qualified to be $O(1)$? No, since in general the exact value of this complexity cannot be set uniformly for all~$x$, except, obviously, when the delimiter reaches the end of the spectrum,~$K(S) \eun K(x)$, with~$S=\{x\}$.
Therefore, the tail of the $\Lambda$-profile is not smooth, since unlike with the $\mathcal L$-profile, small deeper drops may meaninglessly occur. Indeed, these may simply be an artifact of a model $S\ni x$ with larger $K(S)$, but with smaller $K(K(S)\given x, K(x))$.

Hence, a parameter-free notion of sophistication should 
accommodate the fact that~$K(S)$ is in general completely independent from the algorithmic information of~$K(x)$. 
For instance, in the proof of Proposition~\ref{thmlslamcond}, where a model of~$x$ was built from a program that computes~$x$, the length of the two-part description was large enough to encompass the complexity of the delimiter between each part of the description. In fact, this can be seen from Equation~\eqref{eqsuffstatcond}, which reduces to
\bes
K(A)+ \log |A| \lun K(x) + K(i) ~~~~~\text{with}~~~~i \elog K(A) \,,
\ees
if $z= \epsilon$ (no auxiliary information) and $\ell = K(x)$ (build the shortest two-part description from the shortest program). 
Therefore as a candidate for a parameter-free sophistication, one could take
\bes
\min_{S\ni x} \{K(S) \st K(S) + \log|S| \leq K(x) + K(K(S)) + O(1) \} \,,
\ees
which is guaranteed to be witnessed early enough by the two-part description built in the proof of~\ref{thmlslamcond}, for appropriate choice of Proposition~$O(1)$. However, if we are to rely on \emph{the proof of the equivalence} between $\Lambda$ and $\chaitin$ to define sophistication without parameters, we might as well rely on \emph{the equivalence itself}. Like~$\mathcal L$, and unlike~$\Lambda$, $\chaitin$ has a smooth, constant tail of profile, 
 which enables a meaningful definition at $0$ bits of significance.

I then define the \emph{parameter-free sophistication},
 and its conditional homologue, as
\beas
\soph(x) &\deq& \min \{i \st (i,K(x))\in \chaitin_x \} \\
\soph(y \given z) &\deq& \min \{i \st (i,K(y \given z) )\in \chaitina_{y\giv z} \} \,.
\eeas
The unconditional version coincides with Bauwens's~\cite{bauwens2010computability} $m$-sophistication\footnote{With the \emph{a priori} probability as a universal semi-measure.} $k_0$, and within logarithmic precision, with $\bdepth_0$. The conditional version, however, follows~$\chaitina_{y\giv z} \elog \Lambda_{y\giv z}$ instead of $\mathcal L_{y \giv z}$.

With these definitions at hand, the results of Section~\ref{secgap} imply that if $x$ and $y$ have $\chaitin$-profiles with sharp finish,
\bea \label{eqsophpair}
\soph (x,y) 
&\eun& \max \left \{\soph (x),~ \spit_x \left (\soph (y \given x )\right )\right \} \\
&=& \max \left \{\soph (x),~ \soph (y \given x ) + \hd_x \left ( \soph (y \given x )\right )\right \} \,. \nonumber
\eea



Recall the example of Section~\ref{secfirstanti}, showcasing an antistochastic string $z=xy$.
The gap between the conditional profiles illustrated in Figure~\ref{figdiffprof} can now be understood in terms of the halting materialization distribution $H_x(i)$, evaluated at $i=k-|x|$, which consistently amounts to $|x|$. Indeed, from~$x$ and~$\elog k-|x|$ bits of advice (taken from $y$), Milovanov's reconstruction of $z$ can be performed. By the shape of $\mathcal L_z$, such a short program must run for at least a busy running time of~$k$, which is long enough to stabilize~$\glog k$ bits of $\Omega$. 
 Therefore $\spit_x(k-|x|) \elog k$ so~$\hd_x(k-|x|)\elog|x|$. 


Finally, the parameter-free depth and sophistication correspond to the first coordinate of the ``bottom left corners'' of each profile displayed on Figure~\ref{figdiffprof}. Recalling that $\mathcal L_x \elog \Lambda_x$ and $\mathcal L_{x,y} \elog \Lambda_{x,y}$, one finds
\bes
\begin{array}{cccl}
\bdepth_0 (x, y) \elog k \discom& \bdepth_0 (x) \elog 0 \,,&&  \bdepth_0 (y \given x) \elog k \,, \\
\soph (x,y) \elog k \discom& \soph (x) \elog 0 &  \text{and}~~~ &
\soph (y \given x) \elog k - |x|\,. 
\end{array}
\ees
And the established relations~\eqref{eqdepthpair} and~\eqref{eqsophpair} are easily verified.


\section{Conclusions}

The goal has been reached. Thanks to the time profile chain rule of~$\S\ref{sec:crtp}$, the busy beaver depth of a pair~$\bdepth_0(x,y)$ can be expressed in terms of~$\bdepth_0(x)$ and~$\bdepth_0(y\given x)$, simply as their maximum. Had the equivalence of depth and sophistication been carried over by the relative case, it would have been straightforward to formulate a sophistication analogue. The nature of the gap between relative depth and relative sophistication was enlightened in the detour of~\S\ref{secgap}. Best journeys have detours;  
it turns out that this gap reveals more subtle structures in a string~$x$ than those expressed by the Kolmogorov structure function, equivalently represented by~$\Lambda_x$ or~$\mathcal L_x$. In fact, the halting materialization distribution~$\hd_x$ expresses the ability --- or the difficulty --- for~$x$ to solve the halting problem from advices of increasing size.

The antistochastic string~$z$ --- and pieces~$x$ and~$y$ of it --- served as an anchor throughout the paper. Although~$x$ has the same Kolmogorov structure function as any incompressible string, 
 its halting materialization distribution~$H_x$ is very different from that of typical strings: it knows about the halting problem --- and in a somewhat peculiar way. With not enough bits of advice, $x$ is useless to solve \emph{any} halting problem. However, with a large enough advice, its irreducible halting information, \emph{i.e.}, all of its algorithmic information, becomes useful. For more on antistochastic strings, see \S\ref{anholo}.
%

\nit{The Irrelevant Oracle Problem}~\cite{muchnik2010stability}. 
\nopagebreak

\noi  
The gist of that problem can be formulated as follows.
 From a pair of strings~$(x,y)$,~$c$ bits of \emph{common information can be extracted}, for a threshold $t$, if there exists $\gamma$ such that
\bes
K(\gamma \given x) < t \discom K(\gamma \given y) < t \disand K(\gamma) \geq c \,.
\ees
 Assume that~$I(\langle x,y \rangle \colon z) \elog 0$. Can this~$z$ (an apparently irrelevant oracle) help to extract common information between~$x$ and~$y$, \emph{e.g.}, by altering the values $t$ and~$c$ in the relativized case? Muchnik and Romashchenko~\cite{muchnik2010stability} have provided a negative answer when~$x$ and~$y$ are \emph{stochastic} strings, that is, their $\Lambda$ profile contains as many points as possible. But the general case is still open. Can the halting materialization distribution find an application to the problem? 

%


\nit{Depth from Expectation.}

\nopagebreak
\noi The logical depth of~$x$ is defined as the running time of its most probable programs, namely, the shorter ones. This allows us to ignore the fast but long programs, 
such as the ``$\mathtt{Print}~x$'' program. But if such an origin is anyways algorithmically improbable, why not defining the logical depth as the \emph{expected} running time of the computational origines of~$x$; with expectation taken over the algorithmic probability? Something like
\bes
\sum_{p:\U(p)=x} 2^{-|p|} \rt(p) ~?
\ees
It is a nice try, but it makes no sense since this sum diverges for all~$x$. Indeed,
 there exists infinitely many programs~$q$ that purposefully run for much longer than~$2^{|q|}$ steps before producing~$x$.

However, thanks to the busy badger renormalisation, this expectation interpretation of the logical depth can be brought to life. 
I Define the \emph{expected time on the $\Omega$ clock} as
\bes
\mathbb E (\theta_x) \deq \sum_{p:\,\U(p)=x} 2^{-\len p}\theta(p) \,,
\ees
which can be shown to converge for all $x$. It suffices to show that it converges for halting programs. Indeed,
\bes
\mathbb E (\theta_{^{_{^\halts}}}) \deq \sum_{p:\,\U(p)=\halts} 2^{-\len p}\theta(p)
\leq \sum_{\theta} 2^{-\theta} \theta 
= 2 \,.
\ees
The inequality comes from reorganizing the sum, and noticing that the mass of programs running in time $\theta$ or slower on the $\Omega$ clock is less than $2^{-\theta}$, otherwise the value of~$\Omega$ would be contradicted.
This meaningful notion of logical depth as expected running time could perhaps be connected to existing concepts, such as~$\bdepth_0(x)$, which could enhance the justification of its use as the parameter-free depth.




\nit{Programs as Ideas}
\nopagebreak

\noi In Ref.~\cite{bedard2019emergence}, Geoffroy Bergeron and I suggested that the notion of emergence could be associated with the existence of strings that display many drops in their structure function, or in their~$\Lambda_x$ profile. 
Algorithmic models that witness a drop can be thought of a~\emph{new idea}, or a new way to explain the data~$x$. Understanding this concept from the time profile~$\mathcal L_x$ perspective, one finds that those new ideas are equally expressed by programs. The fast but long ``$\mathtt{Print}~x$'' program expresses something radically different from the slow but short~$x^*$. In the middle, everything is possible for some strings thanks to ``All shapes are possible''.  
In particular, there exists a string that admits a very slow~$x^*$ and a very fast program~$p_{\text{scoop}}$, of length that exceeds~$K(x)$ only by an additive logarithmic term...
%


\nit{Algorithmic Randomness in the Universe.}
\nopagebreak

\noi Preeminent physical theories indicate that the Universe originated in a simple state, and has ever since followed algorithmically simple laws. Through a lengthy computation of 14 billion years on what could be thought of as the most powerful computer of the Universe --- the Universe itself --- interesting, non-trivial, deep structures emerged. This is the essence of logical depth.

But what superficially appeared as an easier question might in fact remain a puzzle: how can incidental randomness --- genuine algorithmic randomness --- come about from a simple ``computable'' Universe? 
I see two elements of a tentative answer. First, the only kind of such algorithmic randomness that could be generated is halting information. And it will prosaically arise in time, as any increasing numbers solve ever more halting problems\footnote{This vision is in sharp contrast with Levin's who does not believe that strings with significant mutual information with the halting problem could exist in the world~\cite{levinprivate, vv2002}.}.

Second, what we may think to be fragments of disorder, genuine incidental randomness independent of $\Omega$, may in fact only be pieces of antistochasticity. 
In surface, they seam to be useless noise, but may in fact encode, holographically, the truths about the Universe, \emph{i.e.}, halting information~\cite{chaitin2007}.
This holographic encoding of such deep facts may explain what Deutsch~\cite{deutsch2011beginning} refers to as ``[o]ne of the most remarkable things about science'', namely, ``the contrast between the enormous reach and power of our best theories and the precarious, local means by which we create them.''


\section*{Acknowledgements}
My work is supported by Canada's Natural Sciences and Engineering Research Council (NSERC). 
I am grateful to Charles H. Bennett, Geoffroy Bergeron, Gilles Brassard, Xavier Coiteux-Roy, Samuel Ducharme and Pierre McKenzie for fruitful discussions.
I also wish to thank the Institute for Quantum Optics and Quantum Information of Vienna, in particular, Marcus Huber's group, for a warm welcome and inspiring discussions. 
Last but not least, I am grateful to the veranda, at \emph{Le domaine du pin solitaire}, where this work has been elaborated.

\newpage
\section{Appendix: Holographic Reconstruction from Time Considerations} \label{anholo}


I comment briefly on Milovanov's holographic reconstruction understood by time considerations. Consider as before an antistochastic string~$z=ab$ of length~$n$ and complexity~$k\leq n$ and let~$\len a=k$.
Because of its length, $K(a) \llog k$; but also, $K(a) \glog k$. In fact, 
running $a^*$ and concatenating it to $b$ is one way to compute~$z$, which is 
 of length $K(a) + n - k$ and runs for at most $\bb(K(a))$. This contradicts the time profile~$\mathcal L_z$ unless $K(a) \elog k$. This also means that $a^*$ has busy running time~$\elog k$, namely, the same as running time as $z^*$.

\nit{Claim:} There are at most $2^{s + O(\log n)}$ programs of length $\leq k$ that halt after $\bb(k-s)$ steps\footnote{This is shown in Ref.~\cite[Proposition 13]{vereshchagin2017algorithmic}. Otherwise, it can be understood from the closeness between the busy beaver and the busy badger, and that if too many programs are left to halt, the sum~$M$ would overshoot~$\Omega$.}. 


So letting $s= O(\log n)$, $z^*$ (and so $z$), can be found from an $O(\log n)$ advice if~$\bb(k)$ --- or $\halting{\leq k}$ --- is known, because~$z^*$ is known to be in the last~$2^{O(\log n)}$ halting programs. Therefore, the antistochastic string $z$ becomes simple if the halting problem is solved, which is what $a$ is for.
In fact, from $a$, the logarithmic advice is $K(a)$ permits to find~$a^*$, and by its running time compute~$\bb(k)$ or~$\halting{\leq k}$. 

Since antistochastic strings know so much about the halting problem, the halting problem knows so much about them, making them simple! 
This is the essence of the holographic idea. Any piece of information that solves $\halting{\leq k}$ renders $z$ simple to determine, because one can now start specifying strings from the end of the enumeration. And the particularity of the $\mathcal L_z$ profile ensures that any piece of it that is long enough can be use to determine $\halting{\leq k}$ from a logarithmic advice. 

\bibliographystyle{plain}
\bibliography{refs.bib}
\end{document}